\documentclass[12pt]{article}

\usepackage{amsmath}

\usepackage[sectionbib]{natbib}
\usepackage{bibunits}
\defaultbibliographystyle{chicago} 
\defaultbibliography{bibliography}
\usepackage{algorithm, epigraph, varwidth}
\usepackage[noend]{algpseudocode}

\usepackage[dvipsnames]{xcolor}
\usepackage{amssymb,amsthm}
\usepackage{booktabs}
\usepackage{hyperref}
\usepackage{makecell}

\date{\today}

\hypersetup{
    colorlinks =true,
    allcolors =NavyBlue,
    allbordercolors =white,
}

\definecolor{approvefill}{RGB}{217,232,223}
\definecolor{rejectfill}{RGB}{221,227,234}
\definecolor{ambiguousfill}{RGB}{241,227,191}
\definecolor{approveedge}{HTML}{3E8264}
\definecolor{rejectedge}{HTML}{708696}
\definecolor{ambiguoustext}{HTML}{A88442}
\definecolor{rejecttext}{HTML}{486171}
\definecolor{pointcolor}{HTML}{182533}

\usepackage{etoolbox}
\makeatletter
\patchcmd{\hyper@makecurrent}{%
    \ifx\Hy@param\Hy@chapterstring
        \let\Hy@param\Hy@chapapp
    \fi
}{%
    \iftoggle{inappendix}{
        \@checkappendixparam{chapter}%
        \@checkappendixparam{section}%
        \@checkappendixparam{subsection}%
        \@checkappendixparam{subsubsection}%
        \@checkappendixparam{paragraph}%
        \@checkappendixparam{subparagraph}%
    }{}%
}{}{\errmessage{failed to patch}}

\newcommand*{\@checkappendixparam}[1]{%
    \def\@checkappendixparamtmp{#1}%
    \ifx\Hy@param\@checkappendixparamtmp
        \let\Hy@param\Hy@appendixstring
    \fi
}
\makeatletter

\newtoggle{inappendix}
\togglefalse{inappendix}

\apptocmd{\appendix}{\toggletrue{inappendix}}{}{\errmessage{failed to patch}}

\usepackage{graphicx}
\theoremstyle{plain}
 \newtheorem{theorem}{Theorem}
 \newtheorem*{theorem*}{Theorem}
 \newtheorem{lemma}{Lemma}
 \newtheorem{proposition}{Proposition}
 \newtheorem*{proposition*}{Proposition}
 \newtheorem*{lemma*}{Lemma}
 \newtheorem*{corollary*}{Corollary}
 \newtheorem{corollary}{Corollary}
 \newtheorem*{conjecture*}{Conjecture}
 \theoremstyle{definition}
 \newtheorem{example}{Example}
\usepackage[left=1.25in, right=1.25in,top=1.25in,bottom=1.25in]{geometry}

\usepackage{amssymb}
\usepackage{bbm}
 \usepackage{mathrsfs}
 \usepackage{tikz, tkz-euclide}
 \usetikzlibrary{shapes.misc, positioning, arrows,decorations.markings}
 \usetikzlibrary{calc}
 
 \usepackage{graphicx}
\usepackage{caption}
\usepackage{subcaption}

\usepackage{pgfplots}
\usepgfplotslibrary{fillbetween}
\pgfplotsset{compat=1.18}

\pgfdeclarelayer{background}
\pgfsetlayers{background,main}
\tikzset{
    expand bubble/.style={
        preaction={draw,line width=10.4pt},
        white,fill,draw,line width=10pt,
    },
}

\newlength\Radius
\def \1{\ensuremath\mathbbm{1}}

\begin{document}

\title{Interpersonally Comparable Utility\footnote{We are grateful to Chris Chambers,  Yusufcan Masatlioglu, Axel Niemeyer, Luciano Pomatto, Antonio Rangel, Kota Saito, Larry Samuelson, Charlie Sprenger, and Omer Tamuz for helpful discussions over the course of this project, as well as various seminar audiences.} }
\author{Peter Caradonna\footnote{Division of Humanities and Social Sciences, Caltech.  Email: \href{mailto:ppc@caltech.edu}{\nolinkurl{ppc@caltech.edu}}.}  $\,\!\!$ and Zachary Raines\footnote{Sustaining Engineering Group, Canonical USA Inc. Email: \href{mailto:dev@zmraines.com}{\nolinkurl{dev@zmraines.com}}.}}
\maketitle

\begin{abstract}
We develop a theory of measurement scales for utility functions, based on a generalization of the notion of numeraire commodity. Every sufficiently well-behaved utility is denominated in some scale of this form, and conversely, the choice of a compatible scale unit uniquely identifies utilities up to an additive constant.  We define a profile of utilities to be interpersonally comparable precisely when they are denominated in a common measurement unit. We study when profiles of comparable utilities exist, as well as how a planner ought to aggregate them, and provide applications to social choice and welfare economics. 
\end{abstract}

\epigraph{\emph{``If we want to compare utilities between two people much the same problem exists $\ldots$ we do not seem to have any `outside thing' which can be measured by both persons to ascertain the relation between the units."}}{---\cite{luce1957games} p.\ 34.}

\section{Introduction}

The question of how to obtain cardinally meaningful, and interpersonally comparable, measures of individual well-being is foundational to much of modern economics.\footnote{For example, cost-benefit analysis weighs the gains of some individuals against the losses of others (\citealt{chipman1976scope,chipman1980compensating, schlee2023money}), optimal taxation considers reform incidence across households (\citealt{guesnerie1977direction, tirole1981tax, saez2016generalized,sher2024generalized}); social choice compares both utility levels, and differences, across individuals (\citealt{sen1970collective, sen1974informational, dAspremont1977, roberts1980interpersonal, roberts1980possibility}).} However, despite the centrality of these notions, they have frustrated nearly all attempts to give them concrete, formal meaning within the standard framework of economic theory.\footnote{For an excellent overview, see \cite{hammond1990}.} This has given rise to a broad consensus that any methodology for systematically considering such comparisons must necessarily make use of normative, external judgments.\footnote{E.g.\ \cite{robbins1938interpersonal, arrow1950difficulty, samuelson1954pure,sen1999possibility}. See also \cite{FleurbaeyHammond2004} for further discussion.}

The objective of this paper is to develop a theory of comparable utilities, relying only on ordinal preference information.\footnote{That is, on the basis of ordinal preferences defined over pure outcomes. In particular, we do not consider `extended' preferences over outcome-individual pairs, e.g.\ \cite{suzumura1996interpersonal}.} Rather than emphasize any particular choice of representation, we build on work from \cite{caradonna2024least} to introduce a notion of a `measurement unit' for utilities. This allows us to formally speak of which possible units a given utility function is denominated in, without constraining the values it assigns to particular outcomes.

This, in turn, provides a natural criterion for interpersonal comparability. We define a profile of utilities to be interpersonally comparable precisely when there exists a common unit in which every representation is denominated. While each utility in a profile always independently admits a measurement unit, an arbitrary profile need not admit any common one. Our main results characterize when a profile of utilities admits such a common unit, supply axiomatic foundations for the utilitarian aggregation of comparable utilities, and describe the jointly comparable representations of profiles of ordinal preferences.

Critically, this is not a positive theory of utility---we make no claim that comparable profiles reflect the literal felicities experienced by each individual. Nor does the presence of a common unit, absent additional judgments, imply the existence of an ethical imperative for a planner to treat equal measurements equally. Instead, our results provide a benchmark for when it is justifiable, or \emph{rationalizable}, for a planner to treat utility differences, across both pairs of outcomes and individuals, as quantified on a common scale. We illustrate this below.

\begin{example}\label{introex}
    Suppose society consists of $S$ individuals with preferences over bundles of two commodities, $x = (x^1,x^2) \gg 0$, that can be represented by a quasilinear utility function:
    \begin{equation}\label{introutil}
        u^s(x^1,x^2) = v^s(x^2) + x^1.
    \end{equation}
    Let $\varphi_t(x) = x + (t,0)$ denote the transformations which add $t$ units of the numeraire commodity, whenever feasible.\footnote{Formally, when $t$ is negative, we require that $(x^1+t, x^2)$ remain component-wise positive.} Such transformations are objectively meaningful, as they are defined without reference to any preference. Moreover, every utility \eqref{introutil} satisfies the functional equations:
    \begin{equation}\label{introuu}
        u^s\big(\varphi_t(x)\big) = u^s(x) + t,
    \end{equation}
    for every bundle $x$ and feasible quantity $t$. Thus, for every individual, a utility difference of $t$ has the same interpretation: the value of $t$ additional units of the horizontal-axis commodity. This common scale is not obtained by prescribing particular utility levels to specific allocations. Instead, it is supplied by the exogenously-defined family of transformations $\{\varphi_t\}$. In such circumstances, we regard the utilities $u^s$ as interpersonally comparable.\hfill $\blacksquare$
    \end{example}

More generally, interpersonal comparisons require a common, external scale against which utility differences can be jointly calibrated. By relying on general transfer functions, rather than physical commodities, we are able to define numeraire-like structure more broadly, supplying precisely such a scale.

When society's preferences are over a space of social states or allocations $\mathcal{X}$, we define a `virtual' commodity simply as a collection of transfer functions, $\varphi_t$, each mapping $\mathcal{X} \to \mathcal{X}$, satisfying:
\[
    \varphi_0(x) = x \quad \textrm{ and } \quad \varphi_{t'}\big(\varphi_t(x)\big) = \varphi_{t' + t}(x),
\]
whenever the relevant transfers are feasible.\footnote{More generally, we allow $\varphi_t(x)$ to not be defined for all $x,t$ pairs, to reflect that certain transfers may not be feasible given bounds or other constraints on the space of states $\mathcal{X}$.} We interpret each $\varphi_t$ as augmenting a given state, or allocation, with $t$ additional units of an abstract good. Such a commodity defines a utility unit for a representation $u$ whenever:
\[
    u\big(\varphi_t(x)\big) = u(x) + t,
\]
for every feasible $x$ and $t$, and we define a profile $(u^1,\ldots, u^S)$ to be interpersonally comparable when it admits a common utility unit. Because the transforms $\varphi_t$ are defined over social states, without reference to any preference or utility, they provide a common, objective scale of measurement.\footnote{Throughout, `objective' refers to the specification of the functions $\varphi_t$ as physical transforms of the allocation space. Their suitability as a numeraire unit depends on the preference profile.} In the language of \cite{luce1957games}, they supply the ``outside thing," against which each utility can be measured, to ascertain the equivalence of their units.

Our first main result shows that monotonicity, and a simple `no wealth effects' axiom, characterize the set of virtual commodities that are utility units for some representation of a given preference. Conversely, we show \emph{every} utility representation of a preference admits a unit of this form, which uniquely determines it up to an additive constant. Thus, virtual commodities provide an operational description of the measurement unit implicit in any choice of cardinal representation.

This also highlights that our notion of comparability imposes material restrictions on profiles of utilities. If a profile contains two utilities for the same preference, that differ by a nonlinear, monotone transform, no common utility unit can exist. More generally, we show a simple geometric condition on the gradients of a utility profile completely characterizes comparability, providing a straightforward test for practical applications.

We then consider how a social planner ought to aggregate such utilities. Restricted to the domain of comparable profiles, we show a planner who uses the cardinal measurements, and interpersonal comparisons, that can be obtained on the basis of a common utility unit, but who otherwise satisfies the axioms of \cite{Arrow1951}, is necessarily utilitarian. Notably, our measurement results provide the informational basis for the cardinal and interpersonal comparisons this characterization relies on, rather than requiring us to assume their validity externally. Thus, we provide an underlying, theoretical foundation for the \emph{unit comparability} assumptions widely used in social choice, e.g.\ \cite{sen1970collective, dAspremont1977, roberts1980interpersonal, roberts1980possibility, BossertWeymark2004}.

Finally, we consider the problem of whether a tuple of ordinal preferences admits any comparable profile of representations, and if so, how many. In many instances, a unique profile of comparable utilities exists, up to normalization. In such cases, comparability, and the utilitarian axioms, determine a unique social ranking from ordinal preferences alone.\footnote{See, e.g., \autoref{uniqueex}.} However, even when multiple comparable profiles exist, the rankings agreed upon by the set of comparable-utilitarian social preferences can often provide sharp welfare predictions, even between non-Pareto ranked outcomes.\footnote{See, e.g.\ \autoref{robustex}.} In this event, our results clarify the additional economic or normative considerations needed to select from the set of admissible orderings.

When multiple comparable profiles exist, each corresponds to a different choice of common, utility unit in which to measure well-being. For profiles of expected-utility preferences, this multiplicity provides a unifying perspective on several well-known theories of aggregation. We show utilitarian aggregation of von Neumann-Morgenstern utilities \`{a} la \cite{harsanyi1953cardinal, harsanyi1955cardinal}, relative utilitarianism (\citealt{dhillon1999relative}), and the Nash social welfare function (\citealt{kaneko1979nash}), all correspond to utilitarian aggregation of comparable profiles, differing only in the choice of measurement unit. Moreover, for sufficiently rich profiles, we prove \emph{every} interpersonally-valid utilitarian social preference takes one of these three forms, or a simple, parametric combination of them.

This suggests that, despite the broad consensus that the von Neumann-Morgenstern utilities (or their normalized logarithms) have no special relevance for welfare purposes (\citealt{friedman1948utility, Arrow1951}), these representations actually occupy a privileged position, as the sole comparable utilities for profiles of such preferences. This provides a new perspective on a longstanding critique of Harsanyi's theorems, which asserts his reliance on expected utilities is unjustified. See, e.g., \cite{weymark2005measurement,roemer2008harsanyi}.

In general, describing the set of comparable profiles of representations (or even determining its non-emptiness), for given preferences, is a difficult problem. However, in the setting of local welfare analysis (e.g.\ \citealt{guesnerie1977direction, milleron1970, arrow1970uncertainty, radner1993note}), we provide a complete characterization. This allows us to completely describe the set of local comparable-utilitarian social preferences a planner might entertain. Our results rely on a suite of results due to \cite{cartan1908sous}, developed as part of Felix Klein's \emph{Erlangen} program (\citealt{klein1872vergleichende}), which are collectively known as Cartan's `method of equivalence.' These techniques have found widespread application in physics, geometry, optimization, and control theory.\footnote{See, e.g., \cite{karlhede1980review, gardner1989method, olver1995equivalence,gallo2005cartan} and references therein.} Nevertheless, due to its more technical nature, we treat this question in the \hyperlink{suppapp}{Supplemental Appendix}.

\section{Related Literature}\label{relatedlit}

The question of how to obtain cardinally meaningful, and interpersonally comparable, utilities is the subject of longstanding literatures, both in economics and philosophy. For excellent overviews, see \cite{hammond1990} or \cite{FleurbaeyHammond2004}, and references within.

Recently, there has been renewed interest in the measurement foundations of models of decision-making. \cite{samuelson2026modeling} calls for developing a theory of cardinal preference intensity that allows both intra-, and inter-personal comparisons to be meaningfully constructed. Our results provide just such a theory, applicable both to the theories of risk and uncertainty emphasized by Samuelson, and beyond.

In an elegant recent contribution, \cite{gonczarowski2024quantifying} axiomatize a measure of social inefficiency which carefully incorporates both cardinal, and interpersonal, comparisons, in order to provide foundations for `price of anarchy' theorems in various economic settings.

\cite{alaoui2025whats} provide a theory that cleanly separates risk and marginal utility by introducing an auxiliary  ``numeraire" axis to their consumption space, for use as a yardstick against which marginal utilities can be measured. While we similarly rely on a suitable form of numeraire, we study the distinct problem of when an appropriate form of numeraire can be defined \emph{endogenously}, within a given consumption space, for a given profile of preferences.\footnote{For a more formal connection between these approaches, see \autoref{quasilinearizationthm}.} \cite{sandomirskiy2024geometry} studies the aggregation properties of homothetic preferences using techniques from convex geometry. \cite{echenique2024utilitarian} show that foundations for distributionally-sensitive social welfare functions can be obtained by aggregating individuals' choices, rather than preferences.

{\bf Social Choice}: \cite{Arrow1951} explicitly argues against the possibility of cardinal, interpersonally comparable utility, leading to his impossibility theorem. This assumption was later revisited (\citealt{sen1970collective}, \citealt{sen1974informational}, \citealt{ dAspremont1977}, \citealt{roberts1980interpersonal, roberts1980possibility}, \citealt{blackorby1984social}), yielding possibility results when various types of interpersonal comparability are exogenously assumed.\footnote{See also, e.g., \cite{demeyer1971welfare} for similar ideas in the context of voting, and \cite{BossertWeymark2004} for an excellent overview of these ideas.} A common critique of this approach is that it is unclear when such comparisons are justified (\citealt{FleurbaeyHammond2004}).\footnote{\cite{FleurbaeyHammond2004} (p.\ $\!\!$1271) write ``In retrospect, however, the social welfare functional approach can now be seen to have several quite serious defects. One is the failure to explain how interpersonal comparisons of utility are to be interpreted, since the informational basis was assumed to be exogenously given."} While these works study the planner's problem when presented with cardinal, interpersonally comparable utilities, it has remained an open question to establish when, and how, such comparisons can be justified. Our measurement-theoretic results provide just such a foundation.

A common approach to motivating interpersonal comparability is to restrict preferences to quasilinear or expected utility form (e.g.\ \citealt{harsanyi1955cardinal, moulin1985egalitarianism, gonczarowski2024quantifying}). From this perspective, our notion of a utility unit may be seen as generalizing quasilinearity to a far wider range of profiles and settings. In the \hyperlink{suppapp}{Supplemental Appendix}, we show comparability is equivalent to the existence of a global change-of-coordinates which renders a profile simultaneously quasilinear, with respect to a common numeraire.\footnote{See  \autoref{quasilinearizationthm}.}

Classical welfare economics often sidestepped matters of comparability by assuming that all consumers share a common underlying preference, possibly differing only in their specific representation (\citealt{vickrey1945measuring, mirrlees1971exploration}). More modern approaches include the use of equivalence scales (e.g.\ \citealt{deaton1980almost,lewbel1989household}) or consideration only of preferences whose behavior admits an exact aggregate, e.g.\ \cite{gorman1953community}.\footnote{For an overview of this approach see Section 6 of \cite{FleurbaeyHammond2004}.} Because we allow for nonlinearities in our generalized numeraires, comparable profiles generally need not have demand behavior consistent with any representative consumer. Instead, comparability allows one to directly work with the utilitarian social preference obtained from the sum of the individual utilities.

{\bf Welfare Economics}: Various approaches exist to constructing normalized utility functions with particular units of measure. \cite{wold1944synthesis} provides a construction of a continuous utility function in commodity space which is normalized to be linear along a particular line segment. Similar ideas appear in the literature on `distance functions' (\citealt{malmquist1953index, deaton1979distance}).\footnote{\cite{debreu1951coefficient} suggests a related measure for quantifying inefficiency.} The practical value of these approaches is discussed in \cite{deaton1980measurement}.

A widely used method in consumer theory relies on the use of reference prices to obtain money metric representations (\citealt{mckenzie1957, samuelson1974complementarity}). \cite{chipman1976scope, Willig1976, chipman1980compensating, chipman1992compensating} study the cardinal meaningfulness of such metrics in the setting of cost-benefit analysis. While empirically and theoretically attractive (e.g.\ \citealt{deaton1980measurement, deaton2002guidelines, schlee2023money}), there is limited justification for treating money metric utilities as interpersonally comparable (\citealt{hammond1990, hammond1995money}).\footnote{Indeed, \cite{samuelson1974complementarity} (p.\ 1266) calls the utilitarian summation of money metrics for welfare analysis a `perversion,' and explicitly cautions `any such temptation should be resisted.'}

{\bf Extensive Measurement}: The theory of extensive measurement (\citealt{Suppes1951, krantz1971foundations}) is perhaps closest to our approach. It seeks to obtain canonical, cardinally meaningful numerical representations for orderings, by leveraging algebraic structure on the space of objects over which measurements are constructed. These ideas have found recent application in social welfare (\citealt{Nebel2023,nebel2024extensive}). Our approach may be viewed as somewhat dual. Rather than rely on exogenous algebraic structure on social states to construct a measurement scale, we study when a particular profile of preferences admits sufficient, common algebraic structure for extensive measurement to be carried out.

{\bf Geometry of Preferences \& Preference Profiles}: Formally, our notion of a virtual numeraire is a particular type of continuous symmetry of a preference (\citealt{tyson2013preference, mantovi2016smooth}). \cite{ok2014topological, ok2021fully} consider symmetries of incomplete preferences defined on groups, and study their extension properties. \cite{caradonna2026revealed} highlight the ubiquity of such symmetries in axiomatic work, and develop a general revealed preference theory for testing them.

An earlier notion of virtual numeraire was proposed by \cite{caradonna2024least}, in the context of revealed preference theory. Our definition extends the one given there by allowing for transfers of sufficient magnitudes to be infeasible, providing a more flexible theory.\footnote{However, the theory developed in \cite{caradonna2024least} allows for non-smooth virtual commodities and is valid in any metric consumption space. Here, we focus on the finite-dimensional smooth theory to allow for the use of geometric techniques to establish existence results. Part (i) of our \autoref{virtnumthm} is also similar to Theorem 1 of \cite{caradonna2024least}. However, its proof relies on different underlying arguments, and it provides a strictly smaller axiom set, relative to \cite{caradonna2024least}, precisely due to this additional analytic structure.} However, our interest here is in the question of existence of common numeraires for preference \emph{profiles}, which is not treated in \cite{caradonna2024least}.

Our main results concern the simultaneous restrictions imposed by preference \emph{profiles}, rendering them distinct from existing work studying the geometry of individual preferences (\citealt{debreu1959topological, debreu1972smooth, mas1985theory, chiappori1999aggregation, chiappori2006micro}).

\section{Measurement of Individual Utility}\label{indutil}

Let $\mathcal{X}$ be a $d$-dimensional, convex, relatively open subset of Euclidean space, representing the possible social states, or {\bf allocations} available to society. We will assume $d \ge 2$. The set of individuals comprising the society in question is denoted $\mathcal{S} = \{1,\ldots, S\}$. A {\bf utility function} is a smooth, critical-point-free map $u: \mathcal{X} \to \mathbb{R}$. Similarly, a {\bf preference relation} $\succsim$ is a complete and transitive binary relation  on $\mathcal{X}$ such that, for some utility function $u$:
\[
    x \succsim y \quad \iff \quad u(x) \ge u(y),
\]
for all $x,y \in \mathcal{X}$. We will use $\succ$ (resp.\ $\sim$) to denote the asymmetric (resp.\ symmetric) components of $\succsim$.

    \subsection{Virtual Commodities}\label{commsect}
    
    We now introduce our core technical construction, virtual commodities. Informally, a virtual commodity consists of a family of transformations, each mapping $\mathcal{X} \to \mathcal{X}$. We interpret each function as augmenting an allocation with some fixed quantity of a (possibly nonlinear) composite commodity.  However, as $\mathcal{X}$ itself may be bounded, not every allocation can be endowed with arbitrarily large quantities of additional commodity.

    A {\bf commodity domain} is an open set $\mathcal{D} \subseteq \mathcal{X} \times \mathbb{R}$ such that, for each $x \in \mathcal{X}$, the set:
    \[
    \mathcal{D}^{(x)} = \{t : (x,t) \in \mathcal{D}\}
    \]
    is an open interval containing zero. Intuitively, the sets $\mathcal{D}^{(x)}$ encode the quantities of commodity that can be feasibly added to, or subtracted from, a given allocation.
    
    A smooth map $\varphi: \mathcal{D} \to \mathcal{X}$ is a {\bf transfer function} if, writing $\varphi_t(x)$ for $\varphi(x,t)$, for all $x \in \mathcal{X}$ both (i) $\varphi_0(x) = x$, and (ii) $t_1 \in \mathcal{D}^{(x)}$ and $t_2 \in \mathcal{D}^{(\varphi_{t_1}(x))}$ implies $\varphi_{t_1 + t_2}(x)$ is well-defined, and $\varphi_{t_2}\big(\varphi_{t_1}(x)\big) = \varphi_{t_1+t_2}(x)$.\footnote{By well-defined, we mean that $(t_1+ t_2) \in \mathcal{D}^{(x)}$.}  Interpreting $\varphi_t(x)$ as $x$ plus $t$ units of the commodity, condition (i) simply requires that adding no units not affect any allocation. Condition (ii) is a mild, path-independence property stating that, when feasible, adding $t_2$ units of commodity, to the allocation consisting of $x$ plus $t_1$ units, yields the same outcome as adding $t_1 + t_2$ units to $x$ all at once.

    A {\bf virtual commodity} is then a pair $(\mathcal{D}, \varphi)$, where $\mathcal{D}$ is a commodity domain, and $\varphi$ a transfer function, which is maximal.\footnote{That is, there exists no commodity-domain/transfer-function pair $(\mathcal{D}', \varphi')$ with $\mathcal{D} \subset \mathcal{D}'$ and $\varphi = \varphi' \vert_{\mathcal{D}}$.} Finally, given any virtual commodity $\varphi$, its {\bf generator} is the vector field:
\[
    \dot{\varphi}(x) =    \frac{\partial}{\partial t} \varphi_t(x) \big \vert_{t=0}. 
\]
The vector $\dot{\varphi}(x)$ may be regarded as the perturbation at $x \in \mathcal{X}$ which corresponds to adding an infinitesimal additional quantity of commodity.

\begin{example}
    Suppose $\mathcal{X} = \mathbb{R}^d_{++}$ consists of all bundles of positive quantities of $d$ consumption goods. Any of these $d$ goods may be expressed as a virtual commodity. For example, the map:
    \[
        \varphi_t(x) = x + (t, 0, \ldots, 0)
    \]
    represents the horizontal-axis good as a virtual commodity.\footnote{The relevant commodity domain is defined by $\mathcal{D}^{(x)} = (-x^1, \infty)$.} Here, applying the transformation $\varphi_t$ to some bundle $x$, augments it with $t$ additional units of commodity one, and the generator $\dot{\varphi}(x) = [1, 0, \ldots, 0]^\intercal$ for all $x \in \mathcal{X}$.
    
    However, virtual commodities are significantly more flexible, and allow for nonlinear transfers. For example, $\varphi_t(x) = e^t x$, with $\mathcal{D} = \mathbb{R}^d_{++} \times \mathbb{R}$, defines a virtual commodity, where $\dot{\varphi}(x) = x$.  Here, `adding' $t$ extra units of commodity instead proportionally scales the bundle $x$ by the scalar $e^t$.\hfill $\blacksquare$
\end{example}

\subsection{Measurement of Individual Utility}

A virtual commodity $(\mathcal{D}, \varphi)$ is a {\bf virtual numeraire} for a preference $\succsim$ if it satisfies the following pair of axioms.
\begin{itemize}
    \item[(N.1)]\hypertarget{n1}{} {\bf Invariance}: For all $x,y \in \mathcal{X}$, and $t \in \mathcal{D}^{(x)} \cap \mathcal{D}^{(y)}$:
    \[
        x \succsim y \quad \iff \quad \varphi_t(x) \succsim \varphi_t(y).
    \]
    \item[(N.2)]\hypertarget{n2}{} {\bf Monotonicity}: For all $x \in \mathcal{X}$, and $t \in \mathcal{D}^{(x)}$:
    \[
        \varphi_t(x) \underset{(\succ)}{\succsim} x \quad \iff \quad t \underset{(>)}{\ge} 0.
    \]
\end{itemize}
Invariance says that adding some common quantity of numeraire to two allocations does not affect the preference between them. It rules out cases in which adding the same quantity of commodity to each causes the preference between them to reverse. Monotonicity says the virtual commodity is a good.

A virtual numeraire allows for constructing measurements of preference intensity in a systematic fashion. For example, if for two alternatives $x$ and $y$ there exists some $t \in \mathcal{D}^{(x)}$ such that $\varphi_t(x) \sim y$, we may regard $y$ as being `$t$ units better than $x$.' Equivalently, the individual would accept no less than $t$ units of $\varphi$ as compensation for exchanging the allocation $y$ for $x$. By (\hyperlink{n1}{N.1}), this measurement does not depend on specifically $x$ or $y$, instead only on the indifference surface each belongs to.\footnote{If $x' \sim x$, $y' \sim y$, and $\varphi_{t'}(x') \sim y'$, then necessarily $t = t'$. (\hyperlink{n2}{N.2}) ensures that $y \succ x$ if and only if $t > 0$.} Moreover, these measurements are additive: if $\varphi_t(x) \sim y$ and $\varphi_{t'}(y) \sim z$, then whenever feasible, $\varphi_{t + t'}(x) \sim z$, see \autoref{nummeasure}.

\begin{figure}
\centering
\begin{subfigure}{0.48\textwidth}
\centering
\begin{tikzpicture}[line join = round, line cap = round]

\draw[dashed, stealth-stealth] (5,0) -- (0,0) -- (0,5);
\draw[thick] plot [smooth, tension=1] coordinates { (.5,2.7) (1.5,1.25) (4,.75)};
\draw[thick] plot [smooth, tension=1] coordinates { (1.5,3.75) (2.5,2.75) (4.5,2)};
\draw[thick] plot [smooth, tension=1] coordinates { (2.15,4.85) (3,3.75)  (4.75, 3.3)};

\draw[densely dotted, -stealth] plot [smooth, tension=.75] coordinates {  (0,1.5) (.6,2.275)  (1.5,3.2)   (2.5, 3.8) (2.85, 4.65)};

\draw[densely dotted, -stealth] plot [smooth, tension=.75] coordinates { (.65,0) (2.25, 1.25) (3,2.5) (3.25, 4) (3.55, 4.45) };

\draw[densely dotted, -stealth] plot [smooth, tension=.75] coordinates { (2.5,0)(3,.815) (3.9, 1.75) (4.2, 3) (4.5,3.75) };

\coordinate [label=left:\scriptsize $x$] (0) at (.6,2.275);
\fill (0) circle [black!90, radius=1.5pt];

\coordinate [label=left:\scriptsize $\varphi_t(x)$] (1) at (1.735,3.375);
\fill (1) circle [black!90, radius=1.5pt];

\coordinate [label=above right:\scriptsize $\,\,\,\varphi_{t'+t}(x')$] (5) at (4.32,3.33);
\fill (5) circle [black!90, radius=1.5pt];

\coordinate [label=below:\scriptsize $x'$] (2) at (3,.815);
\fill (2) circle [black!90, radius=1.5pt];

\coordinate [label=above right:\scriptsize $\varphi_{t'}(y)$] (3) at (3.185,3.655);
\fill (3) circle [black!90, radius=1.5pt];


\coordinate [label=above right:\scriptsize $y$] (4) at (3,2.5);
\fill (4) circle [black!90, radius=1.5pt];

\coordinate [label=above right:\scriptsize $\varphi_t(x')$] (7) at (4.05,2.125);
\fill (7) circle [black!90, radius=1.5pt];

\coordinate [label=left:\scriptsize $z$] (6) at (2.28,4.5);
\fill (6) circle [black!90, radius=1.5pt];

\draw[fill, black!90, radius=1.25pt] (2.5,0) circle;
\draw[fill, white, radius=.6pt] (2.5,0) circle;

\draw[fill, black!90, radius=1.25pt] (.65,0) circle;
\draw[fill, white, radius=.6pt] (.65,0) circle;

\draw[fill, black!90, radius=1.25pt] (0,1.5) circle;
\draw[fill, white, radius=.6pt] (0,1.5) circle;
\end{tikzpicture}
\subcaption{Measurement of preference intensity using a virtual numeraire. }
\label{nummeasure}
\end{subfigure}
\begin{subfigure}{.48\textwidth}
\centering
\begin{tikzpicture}[
    dot/.style={circle,fill,inner sep=1.7pt},
    tri/.style={dashed,thin},
    vec/.style={->,line width=1.2pt,>=Stealth}
]
        \def\s{6} 
        \def\L{1.25}
        \pgfmathsetmacro{\H}{sqrt(3)*\L/2}

        \coordinate (A) at (0,0);                          
        \coordinate (B) at (\s,0);                         
        \coordinate (C) at ({\s/2},{\s*sqrt(3)/2});       

        \coordinate [label=below right:\scriptsize $x$] (0) at (3.3,1.08);
        \fill (0) circle [black!90, radius=1.5pt];

        \coordinate [label=below right:\scriptsize $x_1$] (1) at (1.25,1.08);
        \fill (1) circle [black!90, radius=1.5pt];

        \coordinate [label=below right:\scriptsize $\!\!\varphi_{t_1}(x_1)\,$] (2) at (3.55,2.16);
        \fill (2) circle [black!90, radius=1.5pt];

        \coordinate [label=below right:\scriptsize $x_2$] (3) at (1.55,2.16);
        \fill (3) circle [black!90, radius=1.5pt];

    \draw[tri] (A) -- (B) -- (C) -- cycle;

    \foreach \k in {1,2,3} {
        \pgfmathsetmacro{\y}{\k*\H}
        \pgfmathsetmacro{\xleft}{\y/sqrt(3)}
        \pgfmathsetmacro{\xright}{\s-\xleft}

        \draw[black, thick]
            (\xleft,\y) -- (\xright,\y);
    }

    \begin{scope}
        \clip (A) -- (B) -- (C) -- cycle;

        \foreach \k in {1,2,3} {
            \pgfmathsetmacro{\y}{\k*1.1*\H/1.25}

            \draw[black, densely dotted,
                  rotate around={25:(\s/2,\y)}]
                (-1,\y) -- ({\s+1},\y);
        }
    \end{scope}

    \fill[white] (3.95,3.45) circle [white!90, radius=5pt];

    \coordinate [label=above right:\scriptsize $\varphi_{t_2}(x_2)$] (4) at (3.8,3.24);
        \fill[black] (4) circle [black!90, radius=1.5pt];

    \coordinate [label=above left:\scriptsize $y\!$] (4) at (2.5,3.24);
        \fill[black] (4) circle [black!90, radius=1.5pt];


     \end{tikzpicture}
     \subcaption{Quantifying intensity using a sequence of measurements \eqref{decompseq}.}
     \label{nummeasureseq}
\end{subfigure}
\caption{Virtual numeraires for a preference allow for cardinal measurements of intensity, independently of any utility. In the left figure, the measurements of $y$ over $x$ (i.e.\ $t$) and $z$ over $y$ (i.e.\ $t'$) are additive ($z$ is preferred to $x$ by $t+t'$ units), and depend only on the relevant indifference curves. On the right, no feasible transfer at $x$ yields indifference to $y$, but a sequence of measurements suffices.}
\label{virtnumfig}
\end{figure}

More generally, even when no such $t$ exists, the strength of preference for $y$ over $x$ may still be quantified, instead by considering a sequence of measurements:
\begin{equation}\label{decompseq}
            x \sim x_1, \quad \varphi_{t_1}(x_1) \sim x_2, \quad \cdots \quad  \varphi_{t_{K-1}}(x_{K-1}) \sim x_K, \quad \varphi_{t_{K}}(x_K) \sim y.
\end{equation}
In such cases, we regard $y$ to be $\sum_{k=1}^K t_k$ units better than $x$; see \autoref{nummeasureseq}. Our next proposition shows that, no matter the preference, domain, or numeraire, every pair of allocations are necessarily related by some sequence of this form, and that the sum of transfers along any such sequence is uniquely determined.

\begin{proposition}\label{decompprop}
        Let $(\mathcal{D}, \varphi)$ be a virtual numeraire for $\succsim$. Then for any $x, y \in \mathcal{X}$, there exists a finite sequence $(x_1,t_1), \ldots, (x_K, t_K) \in \mathcal{D}$ satisfying \eqref{decompseq}. Moreover, if $(\tilde{x}_1, \tilde{t}_1), \ldots, (\tilde{x}_L, \tilde{t}_L) \in \mathcal{D}$ is any other such sequence, $\sum_{k=1}^K t_k = \sum_{l=1}^L \tilde{t}_l$.
\end{proposition}

Suppose now that $(\mathcal{D}, \varphi)$ is a virtual numeraire for a preference $\succsim$, which is represented by some utility function $u$. We say that $(\mathcal{D}, \varphi)$ is a {\bf utility unit} for $u$ if, for all $(x,t) \in \mathcal{D}$:\footnote{The restriction to virtual numeraires is without loss of generality. If a virtual commodity satisfies \eqref{uu}, it necessarily satisfies (\hyperlink{n1}{N.1}) and (\hyperlink{n2}{N.2}) for the underlying preference.}
\begin{equation}\label{uu}
    u\big(\varphi_t(x)\big) = u(x) + t.
\end{equation}
In such cases, $\varphi$-denominated intensity measurements coincide precisely with pairwise utility differences under $u$. For example, if $\varphi_t(x) \sim y$ for some feasible $t$, then \eqref{uu} implies $t = u(y) - u(x)$.  Thus, we may interpret $\varphi$ as providing a system of measuring tapes, precisely calibrated to the unit-of-measure implicit in the utility function $u$. However, since $\varphi$ can be defined solely from the physical structure of the allocation space, it provides an \emph{objective} scale, whose measurements remain meaningful without any reference to this utility.

Our first main result establishes two basic facts. The first is that every virtual numeraire for a preference on $\mathcal{X}$ is a utility unit for some choice of representation. Thus, when quantifying preference intensity via numeraire, one never obtains measurements which do not arise as the pairwise differences of some utility.

The second is that, for any preference $\succsim$ and choice of representation $u$, there always exists a virtual numeraire for $\succsim$ which is a utility unit for $u$. Thus, virtual numeraires yield a sufficiently rich theory of measurement to capture the utility differences arising from any choice of representation, for any preference. Together, these suggest that the selection of some virtual numeraire for a preference may be identified with a choice of `util,' or measurement unit, for utility.

\begin{theorem}\label{virtnumthm}
    Let $\succsim$ be a preference on $\mathcal{X}$. Then:
    \begin{itemize}
        \item[(i)] Every virtual numeraire for $\succsim$ is a utility unit for some choice of representation of $\succsim$.
        \item[(ii)] For every representation $u$ of $\succsim$, there exists a virtual numeraire for $\succsim$ which is a utility unit for $u$.
    \end{itemize}
    Moreover, two representations of $\succsim$ share a common utility unit if and only if they differ by an additive constant.
\end{theorem}

\autoref{virtnumthm} shows virtual numeraires satisfy the basic desiderata for any notion of cardinal utility unit. For any preference, fixing a choice of virtual numeraire identifies a unique (up to an additive constant) representation satisfying \eqref{uu}. In this sense, such commodities capture the \emph{measurement unit} for utility, as they identify precisely the missing information which, in addition to an individual's underlying preference, uniquely determines utility up to a constant (\citealt{krantz1971foundations}).

Critically, however, virtual commodities may be defined without reference to any preference. In this sense they are \emph{objective} units of measure, external to the ordinal (or cardinal) assessments made by any individual. This allows pairs of allocations, of the form $\big\{x, \varphi_t(x)\big\}$, to supply precisely the `outside thing,' per \cite{luce1957games}, against which utility functions can be measured to ascertain their units.

Nonetheless, an important caveat is that any fixed utility representation admits an infinity of possible units. For example, when $\mathcal{X} = \mathbb{R}^2$ and $u(x^1, x^2) = x^1 + x^2$, any vector $\lambda = (\lambda^1, \lambda^2)$ with $\lambda^1 + \lambda^2 = 1$ defines a utility unit for $u$, via:
\[
    \mathcal{D} =\mathbb{R}^3 \quad \textrm{ and } \quad \varphi_t(x) = x + t\lambda.
\]
Thus while \autoref{virtnumthm} tells us that every representation of a preference is denominated in some numeraire unit, a given utility does not itself uniquely identify a unit of measure. As a consequence, it may be the case that many different utilities, representing distinct preferences, all share a common underlying unit.

\begin{example}
    Consider a finite set of prizes, $\mathcal{Z}$, partitioned into two sets, $G$ and $B$, reflecting good and bad outcomes. Individuals are assumed to have preferences defined on $\mathcal{X} = \textrm{int}\, \Delta(\mathcal{Z})$, i.e.\ full-support lotteries over these prizes, which can be represented by:
    \[
        u^s(p) = \ln\bigg(\frac{p(G)}{p(B)}\bigg) + v^s(p_{\vert G}, p_{\vert B}),
    \]
    where $p(G)$ and $p(B)$ denote the probabilities of the events $G$ and $B$ under $p$, and $p_{\vert G}$ and $p_{\vert B}$ the probabilities conditional on these events. Thus each utility $u^s$ is the sum of the log-odds of receiving a non-disappointing prize, plus an arbitrary function of the conditional probabilities, $p_{\vert G}$ and $p_{\vert B}$.

    Despite there not existing any physical commodities in this environment, the virtual commodity:
    \[
        \varphi_t^j(p) = \frac{e^{t \cdot \mathbbm{1}_{j \in G}}p^j}{p(B) + e^t p(G)}, \quad \quad j \in G\cup B,
    \]
    satisfies \eqref{uu} for every $u^s$, and hence defines a common utility unit. Here, `adding' commodity shifts mass away from the disappointing prizes in $B$, to the more desirable outcomes in $G$, while preserving the conditional probabilities $p_{\vert G}$ and $p_{\vert B}$.\hfill $\blacksquare$ 
\end{example}

\section{Interpersonal Comparisons of Utility}\label{comparablesect}

Suppose now that $(u^1,\ldots, u^S)$ is a profile of utilities. We will say these utilities are {\bf interpersonally comparable} if there exists a common utility unit, i.e.\ a virtual commodity $(\mathcal{D}, \varphi)$ such that, for all $s \in \mathcal{S}$ and all $(x,t) \in \mathcal{D}$,
\[
    u^s\big(\varphi_t(x)\big) = u^s(x) + t.
\]
When such a commodity exists, each utility $u^s$ may be regarded as being denominated in the same objective measurement unit, the numeraire good $\varphi$.

We interpret the property of being interpersonally comparable as indicating that it is \emph{rationalizable}, or internally consistent, for a social planner or policymaker to treat the profile of utilities as denominated in a common measurement unit. This provides measurement-theoretic justification for comparing the magnitude of the gains of some members of society against the losses of others, as captured by these utilities.

An immediate consequence of this definition is that not all utilities are interpersonally comparable. For example, by \autoref{virtnumthm}, if society contains two individuals with the same underlying preference, but whose utilities differ by a nonlinear monotone transform, the resulting profile cannot be interpersonally comparable. Intuitively, such a pair of utilities cannot be rationalized as being denominated in any common measurement unit.

Thus, simply assuming that the utilities in a profile are cardinally meaningful does not, for example, imply their utilitarian sum is a meaningful measure of welfare. While cardinality ensures \emph{intrapersonal} comparisons of preference intensity are internally consistent, absent a common unit, \emph{interpersonal} comparisons remain apples-to-oranges.

Interpersonal comparability also yields ordinal implications.\footnote{See also \autoref{ordinalsect}.} If society's underlying preferences are too heterogeneous, interpersonal comparability can fail for \emph{every} profile of representations. For example, if at some allocation $x$, the intersection of all individuals' strictly preferred sets $\{x' \in \mathcal{X} : x' \succ^s x\}$ is empty, no virtual commodity can satisfy (\hyperlink{n2}{N.2}) for every preference, and hence no common virtual numeraire can exist.\footnote{This is closely related to various domain and profile restrictions in social choice, see \cite{kramer1973class,chichilnisky1983necessary,dhillon1999relative}.}

Finally, we note that comparability is a cardinal property. When a profile $(u^1, \ldots, u^S)$ admits a common unit, even if an identical, strictly increasing transform is applied to each utility, the resulting profile will generally not remain comparable. The following proposition characterizes the tuples of monotone transforms which robustly preserve comparability.

\begin{proposition}\label{uniquenessprop}
    Suppose $S \ge 2$, and let $f = (f^1,\ldots, f^S)$ be a tuple of smooth, strictly increasing, and critical-point-free transformations. Then $f$ has the property that $(f^1 \circ u^1, \ldots, f^S \circ u^S)$ are interpersonally comparable whenever $(u^1, \ldots, u^S)$ are, if, and only if:
    \begin{equation}\label{affinecommonslope}
        f^s(t) = \alpha t + \beta^s,
    \end{equation}
    for $\alpha \in \mathbb{R}_{++}$ and $\beta^s \in \mathbb{R}$, for all $s \in \mathcal{S}$.
\end{proposition}

\subsection{A Characterization of Comparability}

Suppose the planner observes a profile $(u^1,\ldots, u^S) \in \mathcal{U}$. At any $x \in \mathcal{X}$, let $r^\mathcal{S}(x)$ denote the rank of the $d \times S$ matrix whose columns are the gradients of each $u^s$:
\begin{equation}\label{rankdef}
    r^\mathcal{S}(x)= \textrm{rank} \begin{bmatrix}
        & \vdots & \\
        \cdots & \nabla u^s(x) & \cdots\\
        & \vdots & 
    \end{bmatrix}.
\end{equation}
We say that the profile has {\bf constant rank} if $r^\mathcal{S}(x)$ is constant. This is a mild regularity condition which rules out singularities where the dimension of the space spanned by the normal vectors of the underlying indifference surfaces changes discretely. For any profile, the set of allocations at which $r^\mathcal{S}(x)$ is locally constant is always open and dense in $\mathcal{X}$.\footnote{See, e.g., \cite{lee2002} Theorem 4.14. Thus, for matters of \emph{local} welfare analysis (see \autoref{ecn}) the preconditions of \autoref{cardexistance} are generically satisfied.} 

\begin{theorem}\label{cardexistance}
    Suppose that $(u^1,\ldots, u^S)$ has constant rank. Then the profile is interpersonally comparable if, and only if, for all $x\in \mathcal{X}$:
    \begin{equation}\label{cardconsist}
        0 \not \in \textrm{\emph{Aff}}\bigg\{\nabla u^1(x), \ldots, \nabla u^S(x)\bigg\}.
    \end{equation}
\end{theorem}

\autoref{cardexistance} provides a simple, geometric characterization of when a constant-rank profile is interpersonally comparable: the affine hull of the gradients of each utility must never contain the zero vector. This serves as a straightforward test for verifying the comparability of a given profile.

This highlights that, when consumption is heterogeneous and private, comparability cannot be falsified. For example, if $x = (x^1, \ldots x^S)$, where each $x^s \in \mathbb{R}^{n_s}_+$, and each $\succsim^s$ depends only on $x^s$, then for \emph{any} profile of representations, \eqref{cardconsist} is satisfied.\footnote{Similarly, if $ \dim \mathcal{X} \ge S$, then for a generic profile of preferences, locally about a generic allocation, \eqref{cardconsist} will again be satisfied for any profile of utilities.} However, when there are more individuals in society than components of $\mathcal{X}$, comparability imposes strong restrictions. Crucially, while \eqref{cardconsist} is not satisfied by sufficiently rich profiles of generic utilities, the types of structural hypotheses typically assumed in practice often supply common units.

\begin{example}\label{pigex}
    Suppose there are $L$ commodities, and $\mathcal{X} = \mathbb{R}^{L+1}_{++}$ denotes the set of market states, which we assume are described by common prices $p \gg 0$ and wealth $w > 0$.\footnote{In general, the market state naturally depends on the full vector of individual household incomes, rather than any single $w$. The specification here may then be regarded as studying comparability of individual households within a common, known income bracket. See, e.g.\ \cite{deaton1980measurement}.} Suppose that the indirect utility of each household $s \in \mathcal{S}$ takes the form:
    \[
        v^s(p,w) = \frac{f(w) - g^s(p)}{h(p)},
    \]
    where $f'(w) > 0$, and $h(p) > 0$.\footnote{Throughout, we assume all $v^s$ satisfy the standard demand-theoretic regularity conditions.} Then $\varphi_t(p,w) = \big[p, f^{-1}(f(w) + t \,h(p))\big]^\intercal$ may be seen to define a common utility unit. For example, when:
    \[
        v^s(p,w) = \frac{\ln w - \ln \kappa^s - \ln{a(p)}}{\ln{b(p)} - \ln{a(p)}},
    \]
    where $\kappa^s$ is a household-specific scalar, and $a,b$ are concave and homogeneous of degree one, this corresponds to the PIGLOG preferences underpinning the almost-ideal demand system (\citealt{deaton1980almost}). Likewise, the PIGL system (\citealt{muellbauer1976community, deaton1980almost}):
    \[
        v^s(p,w) = \frac{(w/\kappa^s)^\rho - a(p)^\rho}{b(p)^\rho - a(p)^\rho},
    \]
    where additionally $\rho \neq 0$, becomes comparable once each indirect utility $v^s$ is idiosyncratically normalized to $u^s = (\kappa^s)^\rho v^s$.\hfill $\blacksquare$
\end{example}

In \autoref{quasilinearsect}, we provide a complementary description of comparable profiles. We show, under mild regularity conditions, that a profile $(u^1,\ldots, u^S)$ is comparable, with common unit $\varphi$, if and only if there exists a global change of coordinates on $\mathcal{X}$ which (i) simultaneously renders each $u^s$ quasilinear, with respect to a common numeraire good, and (ii) maps each transfer function $\varphi_t$ into a translation along this numeraire's axis. Thus our notion of comparability may conceptually be regarded as simply requiring quasilinearity, but only up to change of coordinates.

\subsection{Comparability \& Arrovian Aggregation}

How ought a planner synthesize a profile of comparable utilities into an aggregate, social ranking? Relative to much of social choice, which supposes an unrestricted domain (see, e.g.\ \citealt{BossertWeymark2004}), as noted above, comparability restricts the set of allowable profiles. Our objective will be to study the aggregation problem on this restricted domain.

Let $\mathcal{U}$ denote all profiles of utilities $(u^1, \ldots, u^S)$ on $\mathcal{X}$, and $\mathcal{C} \subset \mathcal{U}$ the sub-collection of comparable profiles. A restricted social welfare functional $F$ is a mapping from $\mathcal{C}$ to $\mathcal{P}$, the set of preference relations on $\mathcal{X}$. For fixed $F$ and any $u \in \mathcal{C}$, we denote the preference $F(u)$ via $\succeq_u$, with asymmetric component $\succ_u$.

Following \cite{sen1970collective}, a (restricted) social welfare functional is able to rely on cardinal information about utility differences, both intra- and interpersonally, though not information about utility levels, when it satisfies the following axiom.\footnote{In contrast, the case of \cite{Arrow1951}, where $F$ is able to only rely on ordinal information, can be captured by requiring that for any $(u^1,\ldots, u^S) \in \mathcal{U}$, and any tuple $(f^1, \ldots, f^S)$ of increasing transforms $F(f^1 \circ u^1, \ldots, f^S \circ u^S) = F(u^1,\ldots, u^S)$.}
\begin{itemize}
    \item[(CUC)] {\bf Cardinal Unit-Comparability}:\hypertarget{cuc}{} For any profile $(u^1, \ldots, u^S)\in \mathcal{C}$, and any $\alpha > 0, \beta \in \mathbb{R}^\mathcal{S}$:
    \[
        F(u^1, \ldots, u^S) = F(\alpha u^1 + \beta^1, \ldots, \alpha u^S + \beta^S).
    \]
\end{itemize}

\noindent By definition, for any $u \in \mathcal{C}$, it is meaningful to make comparisons of the form:
\begin{equation} \label{diffcomp}
    u^s(x) - u^s(y) \ge u^{s'}(x') - u^{s'}(y'),
\end{equation}
for any $s,s' \in \mathcal{S}$ and any $x,x',y,y' \in \mathcal{X}$, as by \autoref{decompprop}, the quantities on each side of the inequality can be expressed in terms of the common utility unit, $\varphi$. Conversely, (\hyperlink{cuc}{CUC}) says that when a profile of utilities is scaled by a common factor $\alpha > 0$, or shifted by individual quantities $\beta^s$, the induced (ordinal) ranking over social alternatives must remain unchanged. By \cite{bossert1991intra}, such transformations are precisely those which preserve the inequalities \eqref{diffcomp}. Thus, for a social welfare functional defined over comparable profiles, (\hyperlink{cuc}{CUC}) simply encodes the requirement that the planner use only those measurements which can be constructed using $\varphi$, i.e.\ \eqref{diffcomp}, and nothing more.\footnote{Notably, (\hyperlink{cuc}{CUC}) imposes no restriction that a planner treat equal measurements, across different individuals, equally. For example, the ordering represented by $\sum_s\lambda_su^s$ satisfies (\hyperlink{cuc}{CUC}) for any positive weights $\lambda_s$, including those that may be unequal.}

Apart from (\hyperlink{cuc}{CUC}), we consider a planner who is otherwise Arrovian, that is, who satisfies the social welfare functional analogues of the axioms of \cite{Arrow1951}.\footnote{Formally, \cite{Arrow1951} assumes only non-dictatorship; following \cite{dAspremont1977}, we replace this with the symmetry axiom (\hyperlink{a}{A}).} Our treatment here follows \cite{dAspremont1977}.
\begin{itemize}
    \item[(IIA)] {\bf Independence of Irrelevant Alternatives}:\hypertarget{iia}{} For all $x,y \in \mathcal{X}$, and any $u, v \in \mathcal{U}$, if $u(x) = v(x)$ and $u(y) = v(y)$, then $x \succeq_u y$ if and only if $x \succeq_v y$.
    \item[(SP)] {\bf Strong Pareto}:\hypertarget{sp}{} For all $ u \in \mathcal{U}$ and $x,y \in \mathcal{X}$, $u^s(x) \ge u^s(y)$ for all $s \in \mathcal{S}$ implies $x \succeq_u y$.  If, in addition, at least one comparison is strict, then $x \succ_u y$.
    \item[(A)] {\bf Anonymity}:\hypertarget{a}{} For any permutation $\sigma$ of $\mathcal{S}$, and any $u \in \mathcal{U}$, $F(u^1,\ldots, u^S) = $ $ F(u^{\sigma(1)}, \ldots, u^{\sigma(S)})$.
\end{itemize}

Our next result shows that, restricted to comparable profiles, an Arrovian social planner who satisfies (\hyperlink{cuc}{CUC}) is necessarily utilitarian.\footnote{Note that both (\hyperlink{cuc}{CUC}), via \autoref{uniquenessprop}, and (\hyperlink{a}{A}), remain well-defined restricted to $\mathcal{C}$.}

\begin{theorem}\label{utiliarianthm}
    Suppose $F: \mathcal{C} \to \mathcal{P}$ is a social welfare functional for comparable profiles. Then $F$ satisfies (\hyperlink{cuc}{CUC}), (\hyperlink{iia}{IIA}), (\hyperlink{sp}{SP}), and (\hyperlink{a}{A}) if and only if:
    \[
        x \succeq_u y \quad \iff \quad \sum_{s \in \mathcal{S}} u^s(x) \ge \sum_{s \in \mathcal{S}} u^s(y),
    \]
    for all $x,y \in \mathcal{X}$ and $u \in \mathcal{C}$.
\end{theorem}

Methodologically, the proof of \autoref{utiliarianthm} follows closely that of \cite{dAspremont1977}. The key difference is that, rather than proceeding from an unrestricted domain, we consider only comparable profiles of smooth utilities, necessitating new richness-type arguments, and likewise we require the resulting social preference admit a smooth and critical-point-free representation.

Relative to the social choice literature, the novelty of \autoref{utiliarianthm} is in its use of our notion of comparability to justify the assumption of unit-comparability. Historically, the informational basis for making comparisons of the form \eqref{diffcomp} has been assumed to be exogenously given.\footnote{See, e.g., discussion in \cite{FleurbaeyHammond2004}.} Instead, we provide a measurement-theoretic foundation that endogenously provides such utility differences with operational meaning.

While we have focused on utilitarianism as a means of aggregating comparable profiles, our theory is equally applicable to many other social welfare functionals. For instance, when assessing the desirability of a move to a new social allocation, it may be desirable to incorporate information about the \emph{distribution} of utility changes induced by the transition (e.g.\ \citealt{roemer1996theories}). In such cases, the presence of a common unit similarly allows cardinal and interpersonally-justified comparisons to be rigorously incorporated in matters of distributive justice.

\section{Existence of Common Numeraires}\label{ecn}

Suppose a planner observes only a preference profile $(\succsim^1,\ldots, \succsim^S)$. For any choice of common virtual numeraire, \autoref{virtnumthm} guarantees the existence of a unique profile of comparable representations, up to individual additive constants. By \autoref{utiliarianthm}, these uniquely identify the social preference chosen by an Arrovian planner.\footnote{In particular, the additive constants have no effect on the resulting social preference.} Thus when multiple common numeraires exist, there may be multiple, comparable-utilitarian social preferences.

We consider a social planner who is selecting among small policy shifts from a status quo allocation. Such choices may be viewed as determining the \emph{direction} of reform.\footnote{Local welfare analysis problems of this sort have a long history in economics, including in welfare economics, cost-benefit analysis, and optimal taxation; see, for example, \cite{guesnerie1977direction, milleron1970, arrow1970uncertainty, radner1993note}. Recent contributions include \cite{saez2016generalized}, \cite{schlee2023money}, and \cite{sher2024generalized}.} We will seek to determine whether any, and if so, how many, profiles of comparable representations exist, locally about the status quo, for a given preference profile. This, in turn, characterizes the range of comparable-utilitarian objectives such a planner could entertain, when faced with only ordinal information.

For ease of exposition, we illustrate our results in the special case of profiles which admit a particular kind of locally separable representation. More generally, in \autoref{ordinalsect}, we provide a complete description of the set of common, local virtual numeraires, for general profiles. This fully characterizes all possible comparable-utilitarian objectives for problems of local welfare analysis.

Suppose, then, that $x_0$ is some status quo allocation, and there exists a subprofile of preferences $(\succsim^1, \ldots, \succsim^{d+1})$ such that, for some (hence any) choice of utilities, the gradients of the first $d$ representations are linearly independent.\footnote{By possibly relabeling our original profile, we may suppose this subprofile consists of the first $d+1$ preferences.} We say such a subprofile admits a {\bf locally separable representation} if there exists some neighborhood $\mathcal{N}$ of $x_0$, utility representations $v^1, \ldots, v^d$ and $v^{d+1}$, of $\succsim^1\!\vert_{\mathcal{N}}, \ldots, \succsim^{d+1} \!\vert_{\mathcal{N}}$, and vector $\lambda \in \{-1,1\}^d$, such that:
\begin{equation}\label{seprepglobal}
    v^{d+1}(x) = \sum_{i=1}^d \lambda^i v^i(x),
\end{equation}
for all $x\in \mathcal{N}$. In such cases, we may view the locally defined utility $v^{d+1}$ as an additively separable function, but of the first $d$ utilities, rather than in the original coordinates on $\mathcal{X}$.

When $d = 2$, such a subprofile admits a separable representation precisely when the three families of indifference curves, corresponding to each of the three preferences, satisfy the hexagon condition of \cite{debreu1959topological} on some neighborhood of $x_0$. When $d \ge 3$, the existence of a locally separable representation is characterized by Theorem 1 of \cite{goldman1964note}, applied to $u^{d+1}$ in the coordinates $(u^1,\ldots, u^d)$, where the $u^i$ are arbitrary choices of representation.

Our next result characterizes the set of common, local numeraires for any such subprofile. Since any common numeraire for the full profile must also remain so for any such subprofile, this ensures all possible common numeraires belong to a simple, parametric family.

\begin{theorem}\label{locseptheorem}
    Suppose that $(\succsim^1,\ldots, \succsim^{d+1})$ admits a locally separable representation \eqref{seprepglobal} on some connected neighborhood $\mathcal{N}$ of $x_0$ on which $v(x) = \big(v^1(x), \ldots, v^d(x)\big)$ is smoothly invertible. Then a virtual commodity $\varphi$ on $\mathcal{N}$ is a common virtual numeraire if and only if: 
    \[
        \varphi_t(x) = v^{-1}\bigg( e^{at} \, v(x) + \frac{e^{at}-1}{a} b\bigg) \quad \textrm{ and }\quad  \begin{cases}
            a v(x) + b \gg 0\\
           \lambda \cdot \big(a v(x) + b\big) > 0,
        \end{cases} 
    \]
    for scalars $a, b^1, \ldots, b^d$ and all $x \in \mathcal{N}$.\footnote{When $a = 0$, we define $\varphi_t$ to be $v^{-1}\big(v(x) + bt\big)$.}
\end{theorem}

\autoref{locseptheorem} shows that, when some subprofile admits a locally separable representation, every common, local virtual numeraire is constrained to take on a simple, parametric form.  More generally, when \eqref{seprepglobal} holds on $\mathcal{N} = \mathcal{X}$, it fully characterizes the subprofile's common numeraires, providing a parametric representation for any common numeraire of the full profile. It also provides a simple description, for any valid choice of parameters $(a, b^1, \ldots, b^d)$, of the comparable profile denominated in that choice of unit.

\begin{corollary}\label{compcorr}
    For parameters $(a,b^1,\ldots, b^d)$ defining a common numeraire, the unique, normalized, comparable representations with this unit are given by:
    \[
        u^i(x) = \begin{cases}\frac{1}{a} \ln\bigg(\frac{a\, v^i(x) + b^i}{a v^i(x_0) + b^i} \bigg) & a \neq 0\\
        \frac{v^i(x) - v^i(x_0)}{b^i} & a = 0
        \end{cases},
    \]
    for all $i = 1,\ldots, d+1$, where $b^{d+1} \equiv \lambda \cdot b$.
\end{corollary}

For many classes of preferences, existence of a subprofile admitting such a representation is simply a matter of society containing a diverse enough collection of preferences. In such cases, \autoref{locseptheorem} provides a straightforward description of the set of common numeraires.

\begin{example}\label{cdex}
    Suppose $\mathcal{X} = \mathbb{R}^2_{++}$, $x_0 = (1,1)$, and let $(\succsim^1, \succsim^2, \succsim^3)$ be represented by Cobb-Douglas utilities:
    \[
        \tilde{u}^i(x^1, x^2) = (x^1)^{\alpha^i}(x^2)^{1-\alpha^i} \quad \textrm{ and } \quad \tilde{u}^3(x^1, x^2) = (x^1)^\beta(x^2)^{1-\beta},
    \]
    for $i=1,2$, where we assume that $0 < \alpha^1 < \beta < \alpha^2 < 1$. By direct calculation:
    \[
        \ln \tilde{u}^3 = \gamma^1 \ln \tilde{u}^1 + \gamma^2 \ln \tilde{u}^2,
    \]
    where $\gamma^1 = \frac{\alpha^2 - \beta}{\alpha^2 - \alpha^1}$ and $\gamma^2 =  \frac{\beta - \alpha^1}{\alpha^2 - \alpha^1}$. Thus we may take $v^i = \gamma^i \ln \tilde{u}^i$, for $i = 1,2$, and $v^3 = v^1 + v^2$. By \autoref{locseptheorem}, every common, local virtual numeraire $\varphi$ takes the form:
    \[
        \varphi_t(x) = v^{-1}\bigg( e^{at} \, v(x) + \frac{e^{at}-1}{a} b\bigg),
    \]
    subject to the positivity constraint $a \gamma^i \ln \tilde{u}^i + b^i > 0$ for all $i = 1,2$.\footnote{Here, since $\alpha^1 < \beta < \alpha^2$, $\lambda = (1,1)$ hence the final constraint $\lambda \cdot (av(x) + b) > 0$ is redundant.} Since we have a separable representation globally, \autoref{locseptheorem} characterizes globally-defined numeraires. Now, because $\ln \tilde{u}^i$ takes on every real value, this positivity constraint forces $a = 0$, implying, every globally-defined common numeraire is given by:
    \[
        \varphi_t(x) = \begin{bmatrix}
            e^{c^1 t} \, x^1\\
            e^{c^2 t} \, x^2           
        \end{bmatrix},
    \]
    where $c^1$ and $c^2$ depend on the choice of $b \gg 0$.\footnote{Concretely:
    \[
        c^1 = \frac{(1-\alpha^1)b^2/\gamma^2 - (1-\alpha^2)b^1/\gamma^1}{\alpha^2 - \alpha^1} \quad \textrm{ and } \quad c^2 = \frac{\alpha^2 b^1/\gamma^1 - \alpha^1 b^2/\gamma^2}{\alpha^2 - \alpha^1}.
    \] These may be obtained by explicitly computing the inverse $v^{-1}$.}  By \autoref{compcorr}, the normalized, comparable utilities associated with any such choice of parameters are:
    \[
        u^i = \frac{\gamma^i}{b^i} \ln \tilde{u}^i, \quad i = 1,2, \quad \textrm{ and } \quad u^3 = \frac{\gamma^1 \ln \tilde{u}^1 + \gamma^2 \ln \tilde{u^2}}{b^1 + b^2}.
    \]
    Thus, rather than sum the Cobb-Douglas utilities directly, \autoref{utiliarianthm} calls for a planner to instead maximize their \emph{Nash product}, with weights reflecting the particular choice of measurement unit, as determined by the parameters $b$.\hfill $\blacksquare$
\end{example}

\subsection{Non-Separable Profiles}

When no subprofile admits a locally separable representation, determining the existence of common numeraires is governed by an overdetermined system of partial differential equations.\footnote{See \autoref{pdelemma}.} In general, no solutions to these equations exist unless specific integrability, or compatibility, conditions are satisfied. Using geometric techniques, we characterize the solution sets of such systems in \autoref{ordinalsect}.

This provides criteria not only for determining the set of local numeraires for a given preference profile, but also for obtaining necessary and sufficient conditions on \emph{families} of profiles, for the existence of some tuple of comparable local representations. For example, when $\mathcal{X} = \mathbb{R}^3_{++}$, we show that a profile of preferences represented by utilities of the form:
\[
    u^i(x) = x^i, \quad i = 1,2,3, \quad \textrm{ and } \quad u^4(x) = f(x^1+x^2) + x^3,
\]
for some smooth, increasing and strictly concave $f$, admits a common local numeraire if, and only if, $f$ is additionally the Bernoulli utility of a HARA expected utility function (\citealt{merton1971optimum}).\footnote{Moreover, when non-empty, the set of local numeraires always depends smoothly on exactly three parameters. We prove this in \autoref{ordinalsect}.}

\section{Applications \& Illustrations}\label{exandapp}

\subsection{Aggregation of Expected Utility Preferences}\label{euex}

Let $\mathcal{Z} \subset \mathbb{R}$ denote a finite set of monetary prizes, and suppose that $\mathcal{X}$ is the interior of the simplex $\Delta(\mathcal{Z})$.  Let $(\succsim^1,\ldots, \succsim^S)$ be a profile of monotone expected utility preferences on $\mathcal{X}$. We will denote the highest prize in $\mathcal{Z}$ by $z^*$ and the lowest by $z_*$. Finally, let $(\tilde{u}^1,\ldots, \tilde{u}^S)$ denote any profile of expected utility representations.

The problem of aggregating expected utility preferences has given rise to an extensive literature. One well-known approach, due to \cite{harsanyi1953cardinal, harsanyi1955cardinal}, calls for a planner maximizing a weighted sum of expected utility representations. Relative utilitarianism (\citealt{dhillon1999relative}) calls for instead normalizing the range of each expected utility function to the unit interval, then summing the resulting representations. The Nash social welfare functional (\citealt{kaneko1979nash}) proposes normalizing the expected utility of the worst prize to zero, then maximizing the Nash product of these representations.

In contrast, given a choice of common numeraire, \autoref{utiliarianthm} describes the social preference obtained by aggregating the comparable utilities denominated in this unit. However, when multiple common numeraires exist, this supplies a multiplicity of different objectives for such a planner.\footnote{See \cite{sen1970collective} for a discussion.} As we now show, this multiplicity precisely recovers, and unifies, the above approaches.

Suppose first that $b \in \mathbb{R}^\mathcal{Z}$ is any zero-sum vector such that $p + b \succ^s p$ for some $p \in \mathcal{X}$ and every $s \in \mathcal{S}$. Then any {\bf translation} commodity, $\varphi_t(p) = p + t b$, defines a common numeraire.\footnote{Here, for any lottery $p$, the interval $\mathcal{D}^{(p)}$ is simply:
    \[
        \bigg(\max_{\{i : b_i > 0\}}\bigg(\frac{-p_i}{b_i}\bigg), \min_{\{i : b_i < 0\}} \bigg(\frac{-p_i}{b_i}\bigg) \bigg).
    \]
Since $b$ is zero-sum and not identically equal to zero, the max and min are taken over non-empty sets.} The unique, comparable utilities denominated in this unit are:
    \[
        u^s(p) = \frac{1}{b \cdot \nabla \tilde{u}^s} \tilde{u}^s(p),
    \]
    up to individual additive constants.

In this case, \autoref{utiliarianthm} calls for the planner to maximize precisely a weighted sum of the expected utilities $\tilde{u}^s$, as in \cite{harsanyi1953cardinal, harsanyi1955cardinal}. Harsanyi notes that the social weights in his utilitarian sum are not uniquely determined; here, an identical indeterminacy obtains as a consequence of the fact that there are many choices of translation direction $b$ that yield a common numeraire, and the social weights $(\nabla \tilde{u}^s \cdot b)^{-1}$ depend upon this selection.\footnote{See also \cite{mononen2025observable}.}

When $b$ is proportional to $\delta_{z^*} - \delta_{z_*}$, the comparable profiles remain of expected utility form. However, the normalization imposed by the social weights now forces the range of each $u^s$ to be an open interval of common length.\footnote{Concretely, if $\varphi_t(p) = p + t \varepsilon (\delta_{z^*} - \delta_{z_*})$, for some $0 < \varepsilon < 1$, then:
    \[
        t = u^s\big(\varphi_t(p)\big) - u^s(p) = u^s\big(p + t \varepsilon(\delta_{z^*} - \delta_{z_*})\big) - u^s(p) = t \varepsilon \big[u^s(\delta_{z^*}) - u^s(\delta_{z_*})\big],
    \]
hence the range of each $u^s$ is an open interval of length $1/\varepsilon$, for all $s \in \mathcal{S}$.} As a consequence, the comparable-utilitarian sum called for by \autoref{utiliarianthm} is, up to an affine transform, precisely the objective of a \emph{relative} utilitarian (\citealt{dhillon1999relative}).

Alternatively, defining $\varphi_t(p) = e^{at}p + (1-e^{at})\delta_{z_*}$, where $a > 0$, also yields a common numeraire. Here, $\varphi_t$ corresponds to a {\bf dilation} about the least-desirable prize, $\delta_{z_*}$.\footnote{E.g.\ when  $a=1$, the associated commodity domain is defined by $\mathcal{D}^{(p)} = \big(-\infty, -\ln(1-p_{z^*})\big)$.} Under this new form of numeraire, the comparable profiles are no longer of expected utility form. Instead, they are given by $u^s(p) = \frac{1}{a}\ln\big(\tilde{u}^s(p) - \tilde{u}^s(\delta_{z_*})\big)$, up to individual additive constants. Thus, under this choice of unit, \autoref{utiliarianthm} exactly calls for maximizing the (logarithm of) the Nash social welfare functional (\citealt{kaneko1979nash}).

While these methods of aggregation admit distinct, independent axiomatizations, each also arises directly, as a special case of comparable-utilitarianism, differing only in the choice of implicit measurement unit. In light of \autoref{utiliarianthm}, this provides a novel, unifying axiomatic perspective on these theories. In particular, our results clarify the precise role that interpersonal comparisons, and the specific measurement units used to construct them, play in each.\footnote{This has been the subject of extensive debate, e.g.\ \cite{harsanyi1975nonlinear, sen1976welfare, Sen1977NonLinear, roemer2008harsanyi, weymark2005measurement}.}

In light of this, a natural question is whether there exist new, qualitatively distinct methods of aggregation, corresponding to non-translationary, or non-dilationary, numeraires. For sufficiently rich profiles, \autoref{locseptheorem} implies these two types of numeraire, and their combinations, completely exhaust the possibilities.

\begin{corollary}\label{eucorr}
    Suppose $(\succsim^1, \ldots, \succsim^S)$ is a profile of monotone expected utility preferences containing a subprofile $(\succsim^{s_1}, \ldots, \succsim^{s_{\vert \mathcal{Z} \vert}})$ such that:
    \begin{itemize}
        \item[(i)] $\nabla \tilde{u}^{s_1}, \ldots, \nabla \tilde{u}^{s_{\vert \mathcal{Z} \vert - 1}}$ are linearly independent; and
        \item[(ii)] $\nabla \tilde{u}^{s_{\vert \mathcal{Z} \vert}}$ is a linear combination of the $\nabla \tilde{u}^{s_1}, \ldots, \nabla \tilde{u}^{s_{\vert \mathcal{Z} \vert - 1}}$ with positive coefficients.
    \end{itemize}
    Then every common numeraire can be written as a translation, followed by a dilation.\footnote{Note that when combined with a translation, the dilation parameter $a$ need not be positive, so long as (\hyperlink{n2}{N.2}) is satisfied. For example, the numeraire $\varphi_t(p) = e^{-t}p + (1-e^{-t}) \delta_{z^*}$, which mixes with $\delta_{z^*}$, can be written as a translation in the direction $\delta_{z^*} - \delta_{z_*}$, followed by a dilation with $a = -1$.}
\end{corollary}

A longstanding critique of many methods of aggregation is that their use of specifically expected utility functions is arbitrary and unjustified.\footnote{\cite{friedman1948utility} write ``It is entirely unnecessary to identify the quantity that individuals are to be interpreted as maximizing [their expected utilities] with a quantity that should be given special importance in public policy.'' Similarly, \cite{Arrow1951} writes ``[von Neumann \& Morgenstern's] theorem does not, as far as I can see, give any special ethical significance to the particular utility scale found. For instead of using the utility scale found by von Neumann and Morgenstern, we could use the square." See also \cite{roemer2008harsanyi, weymark2005measurement}.} In contrast, \autoref{eucorr} shows that expected utility representations, and their normalized logarithms, occupy a privileged position: they are the only representations which make up the \emph{comparable} profiles. Notably, this elevation arises from purely measurement-theoretic consideration of the preference profiles themselves, rather than external arguments based around risk.

\subsection{Social Discount Rates}

Heterogeneous time preferences introduce well-known difficulties in problems of collective choice.\footnote{See, e.g.\ \cite{jackson2014present, jackson2015collective, millner2018discounting}.}  In such settings, a particularly natural question is to determine an appropriate aggregate discount rate, that suitably reflects individuals' underlying preferences.\footnote{E.g.\ \cite{weitzman2001gamma, feng2018social, chambers2018multiple, chambers2020pareto}.}

In this section, we consider the implications of comparable-utilitarianism for profiles of exponential discounting preferences, on a domain of common, dated rewards. While any profile of such preferences admits multiple joint numeraires, for any sufficiently rich profile, every comparable-utilitarian social preference is of exponentially-discounted form. In particular, each choice of measurement unit supplies a well-defined social discount rate. This rate turns out to depend on the numeraire unit in a particularly transparent way.

Suppose $\mathcal{X} = \mathbb{R}_{++}^2$, where $x = (c,\tau)$ represents a common payoff of $c$ units of consumption good, delivered at date $\tau$. Society's preferences are represented by utilities:\footnote{The assumption that for fixed $\tau$ utility is linear in $c$ is for convenience. Our results remain valid under any common, smooth consumption utility $f$ where $f > 0$, $f' > 0$, and $f$ is unbounded above.}
\[
    \tilde{u}^s(c,\tau) = c e^{-\delta_s \tau} \quad \textrm{ where } \quad \delta_s > 0.
\]
We suppose society contains at least three individuals, $s_1, s_2,$ and $s_3$, whose respective discount rates satisfy $\delta_* < \delta < \delta^*$, where $\delta_* = \min_{s \in \mathcal{S}} \delta_s$ and $\delta^* = \max_{s \in \mathcal{S}} \delta_s$. Let:
\[
    v^{s_1} = \frac{\delta^* - \delta}{\delta^* - \delta_*} \ln{\tilde{u}^{s_1}}, \quad v^{s_2}(c,t) = v^{s_1} + v^{s_3}, \quad \textrm{ and } \quad v^{s_3}  = \frac{\delta - \delta_*}{\delta^* - \delta_*} \ln{\tilde{u}^{s_3}}.
\]
These monotone transforms of the $\tilde{u}^{s_i}$ form a (globally) separable representation \eqref{seprepglobal} for the subprofile $(\succsim^{s_1}, \succsim^{s_2}, \succsim^{s_3})$. By \autoref{locseptheorem}, the full collection of common virtual numeraires is given by:
\begin{equation}\label{discountnum}
    \varphi_t(c,\tau) = \big(e^{\gamma_1 t} c, \tau - \gamma_2 t\big), \quad \gamma_1 + \gamma_2 \delta_s > 0
\end{equation}
where the vector $\gamma$ satisfies $\gamma_1 + \gamma_2 \delta_s > 0$ for all $s \in \mathcal{S}$.\footnote{This follows from an essentially identical argument to that in \autoref{cdex}, applied to the map $v = (v^{s_1}, v^{s_3})$.} Up to individual additive constants, the comparable utilities are then:
\[
    u^s(c,\tau) = \frac{\ln \tilde{u}^s(c,\tau)}{\gamma_1 + \gamma_2 \delta_s} = \frac{\ln{c} - \delta_s \tau}{\gamma_1 + \gamma_2 \delta_s},
\]
and hence, as a function of measurement unit determined by $\gamma$, the comparable-utilitarian social preferences are represented by:
\[
    U(c,\tau; \gamma) = c e^{-\bar{\delta}(\gamma) \tau}, \quad \textrm{ where } \quad \bar{\delta}(\gamma) = \frac{\sum_{s \in \mathcal{S}} \delta_s / (\gamma_1 + \gamma_2 \delta_s)}{\sum_{s \in \mathcal{S}} 1/(\gamma_1 + \gamma_2 \delta_s)}.
\]
Thus, no matter the choice of common numeraire, the comparable-utilitarian social preference is of the same discounted form, In particular, it is always described by a well-defined {\bf social discount rate} $\bar{\delta}(\gamma)$.

As in \autoref{euex}, common numeraires take on interpretable forms. When $\gamma = (1,0)$, $\varphi_t$ measures well-being in terms of proportional scalings of the common payoff. When utilities are aggregated in this unit, the social discount rate is precisely the \emph{arithmetic mean} of the individual discount rates:
\[
    \bar{\delta}(1,0) = \frac{1}{S} \sum_{s \in \mathcal{S}} \delta_s.
\]
Alternatively, when $\gamma = (0,1)$, measurement is instead quantified in terms of a common advance in delivery date. Here, the comparable-utilitarian social discount rate corresponds to the \emph{harmonic mean} of the individual discount rates:\footnote{Arithmetic and harmonic means of discount rates similarly arise in \cite{millner2018discounting}, through equilibrium considerations distinct from our approach here.}
\[
    \bar{\delta}(0,1) = \frac{S}{\sum_{s \in \mathcal{S}} \frac{1}{\delta_s}}.
\]
As in the case of expected utility profiles, the full set of common numeraires arises as a parametric interpolation of these two types. As $\gamma$ varies across all component-wise nonnegative values, the set of resulting social discount rates is precisely the interval connecting these two means.\footnote{More generally, if $\gamma$ is varied across its full domain, defined by the restrictions $\gamma_1 + \gamma_2 \delta_s > 0$ for all $s \in \mathcal{S}$, the range of possible social discount rates is precisely $(\delta_*, \delta^*)$, i.e.\ the full range of individual discount rates.}

In contrast, while the present-equivalents $\tilde{u}^s$ may be regarded as denominated in units of the consumption good, as the following example illustrates, their sum does not represent any comparable-utilitarian social preference, for any choice of numeraire.\footnote{Informally, $\tilde{u}^s(c,t)$ may be interpreted as the quantity of consumption good that, if received instantly, would make the individual indifferent with $(c,t)$.}

\begin{example}
    Suppose $s= 1,2,3$, with discount rates $(\ln{2}, \ln{3}, \ln{6})$, and let $\tilde{U} = \tilde{u}^1 + \tilde{u}^2 + \tilde{u}^3 = c(2^{-t} + 3^{-t} + 6^{-t})$. Then:
    \[
        \tilde{U}(1,1) = 1 > \frac{35}{36} = \tilde{U}\bigg(\frac{5}{2}, 2\bigg),
    \]
    but:
    \[
        \tilde{U}(1,2) = \frac{7}{18} < \frac{5}{12} = \tilde{U}\bigg(\frac{5}{2}, 3\bigg).
    \]
    Thus the preference represented by $\tilde{U}$ is not time-stationary (\citealt{fishburn1982time}), and hence cannot be of exponentially discounted form, implying it cannot be comparable-utilitarian for any choice of unit.\footnote{More generally, this may be shown to always be true of the sum of the $\tilde{u}^s$.} \hfill $\blacksquare$
\end{example}

Thus, for a comparable-utilitarian social planner, choice of a measurement unit uniquely determines a social discount rate. While different choices generate different rates, the requirement of aggregating specifically comparable profiles imposes substantive restrictions (e.g.\ exponentially discounted form) on social preferences. Our results then characterize the range of admissible objectives, parameterized by this social discount rate, faced by the planner.

\subsection{Robust Welfare Comparisons}

In contrast to \autoref{cdex} and \autoref{euex}, where a multiplicity of common numeraires obtained, many other preference profiles admit a single-parameter family of common numeraires. In such instances, comparability and the aggregation axioms are dispositive, and provide a \emph{unique} social ranking over alternatives.

\begin{figure}
\centering
    \begin{tikzpicture}[scale = .6]
\pgfmathsetmacro{\plotmax}{2.8}
\pgfmathsetmacro{\xstart}{
    (sqrt(\plotmax*\plotmax + 8/\plotmax)-\plotmax)/2
}

\begin{axis}[
    width=8.3cm, height=8.3cm, scale only axis,
    xmin=0, xmax=\plotmax, ymin=0, ymax=\plotmax,
    enlargelimits=false,
    axis lines=left,
    axis on top,
    axis background/.style={fill=ambiguousfill},
    label style={font=\sffamily\fontsize{11}{13}\selectfont},
    xtick={0,1,2},
    ytick={0,1,2},
    scaled ticks=false,
    tick label style={
        font=\sffamily\fontsize{10}{12}\selectfont,
        /pgf/number format/fixed,
        /pgf/number format/fixed zerofill,
        /pgf/number format/precision=1
    },
    tick align=outside, tick pos=left,
    major tick length=3pt,
    axis line style={-,draw=black!70,line width=0.8pt},
    tick style={draw=black!70,line width=0.8pt},
    clip=false, 
]

\addplot[draw=none,fill=rejectfill] coordinates {
    (0,0) (2,0) (0,2) (0,0)
} \closedcycle;

\addplot[
    name path=frontier, draw=none,
    domain=\xstart:\plotmax, samples=300
] {(sqrt(x*x+8/x)-x)/2};
\path[name path=ceiling]
    (axis cs:\xstart,\plotmax) -- (axis cs:\plotmax,\plotmax);
\addplot[draw=none,fill=approvefill]
    fill between[of=frontier and ceiling];

\addplot[draw=rejectedge,line width=0.9pt]
    coordinates {(0,2) (2,0)};
\addplot[
    draw=approveedge,line width=1.2pt,
    domain=\xstart:\plotmax,samples=300
] {(sqrt(x*x+8/x)-x)/2};

\node[anchor=west,text=ambiguoustext,font=\sffamily\fontsize{9}{10}\selectfont]
    at (axis cs:0.07,2.33) {Unit-dependent};
\node[anchor=west,align=left,text=approveedge,
      font=\sffamily\fontsize{9}{10}\selectfont]
    at (axis cs:1.5,2.23) {Improve for\\every unit};
\node[anchor=west,align=left,text=rejecttext,
      font=\sffamily\fontsize{9}{10}\selectfont]
    at (axis cs:0.22,0.41) {Decrease for\\every unit};

\addplot[only marks,mark=*,mark size=2.6pt,
         mark options={draw=pointcolor,fill=pointcolor}]
    coordinates {(1,1)};
\node[anchor=south west,font=\sffamily\fontsize{9}{10}\selectfont,
      xshift=-2pt,yshift=2pt]
    at (axis cs:1,1) {Status quo};

\end{axis}
\end{tikzpicture}
\caption{Aggregate welfare implications under the set of units entertained by the planner in \autoref{robustex}. An allocation in the green region weakly increases welfare under every unit the planner considers; an allocation in the gray region weakly decreases overall well-being in every unit.}
\end{figure}

\begin{example}\label{uniqueex}
Suppose $\mathcal{X} = \mathbb{R}^2_{++}$ represents the possible quantities of a pair of public goods, and let $(\succsim^1, \succsim^2, \succsim^3)$ denote the preference profile represented by:
\[
    \tilde{u}^1 = (x^1)^\alpha (x^2)^{1-\alpha}, \quad \tilde{u}^2 = (x^1)^{1-\alpha} (x^2)^\alpha, \quad \textrm{ and } \quad \tilde{u}^3 = \tilde{u}^1 + \tilde{u}^2 + \kappa \sqrt{x^1x^2},
\]
where $0< \alpha < \frac{1}{2}$, and $\kappa > 0$, $\kappa \neq 2$.\footnote{The restriction that $\kappa \neq 2$ ensures this profile is not separable. If $\kappa = 2$, then $\tilde{u}^3 = \tilde{u}^1 + 2 \sqrt{\tilde{u}^1\tilde{u}^2} + \tilde{u}^2$, and hence $\sqrt{\tilde{u}^3} = \sqrt{\tilde{u}^1} + \sqrt{\tilde{u}^2}$. In this case, \autoref{locseptheorem} would apply, yielding a family of common numeraires depending on more than a single parameter.} We consider a planner who is tasked with deciding between some number of policies, whose impact on society can be described wholly by their impact on the level of provision of the goods $x^1$ and $x^2$. How ought a planner evaluate a (possibly non-Pareto ranked) set of such policies?

To make such comparisons, we first select a common virtual numeraire for society, then sum the resulting comparable utility representations denominated in this choice of unit. This sum then serves to rank the policies under consideration. For the profile $(\succsim^1, \succsim^2, \succsim^3)$, any commodity of the form:
\[
    \varphi_t(x) = e^{\lambda t} x,
\]
where $\lambda > 0$ and $\mathcal{D} = \mathbb{R}^2_{++} \times \mathbb{R}$, defines a common numeraire. Indeed, it can be shown that these are the \emph{only} common numeraires.\footnote{This follows from applying \autoref{commonvnthm} in \autoref{ordinalsect} to asymmetric allocations, to obtain a unique numeraire $\tilde{\varphi}$ off the diagonal, then taking the (unique) smooth extension of $\tilde{\varphi}$ to all of $\mathcal{X}$.} For any value of $\lambda$, the comparable utilities are given by:
\[
    u^s = \frac{1}{\lambda} \ln \tilde{u}^s + \beta^s.
\]
In particular, the social ranking represented by the sum $\sum_s u^s$ is independent of both $\lambda$, and the vector of individual constants, $\beta$. Relative to, e.g., a status quo of $(1,1)$, a policy yielding allocation $x$ is overall welfare-improving precisely when $(x^1x^2)\tilde{u}^3(x) \ge 2 + \kappa$. This provides the planner with a unique criterion for evaluating the welfare impact of any policy, including those Pareto-incomparable to the status quo.\hfill $\blacksquare$
\end{example}

Even absent the strong uniqueness properties of \autoref{uniqueex}, comparability often still provides a means of refining the Pareto ordering. One natural approach is to consider those comparisons robustly agreed upon by the comparable-utilitarian social preferences arising under every, or some normatively-appealing subset, of common numeraires.

\begin{example}\label{robustex}
    Consider again the setting of \autoref{uniqueex}, and suppose society's current status-quo allocation is $(1,1)$. The planner faces a profile of preferences $(\succsim^1,\succsim^2, \succsim^3)$ represented by:
    \[
        \tilde{u}^1 = x^1, \quad \tilde{u}^2  = x^2, \quad \textrm{ and } \quad \tilde{u}^3 = x^1 + x^2.
    \]
    By \autoref{locseptheorem}, every common virtual numeraire is of the form:
    \[
        \varphi_t(x) = e^{at} \begin{bmatrix} x^1\\ x^2\end{bmatrix} + \frac{e^{at} - 1}{a} \begin{bmatrix} b^1 \\b^2\end{bmatrix},
    \]
    for $a, b^1, b^2 \ge 0$. Suppose, given the symmetric distribution of preferences in society, that for equity reasons the planner wishes to consider only virtual numeraires that treat both goods equally: $b^1 = b^2 = b$. For definiteness, we may without loss normalize units by fixing $\dot{\varphi}$ at the status-quo allocation to be $(1,1)$. Thus $b = 1-a$, and $0 \le a \le 1$.

    Let $\Delta W_a(x)$ denote the change in welfare, relative to the status quo, of the utilitarian sum of comparable profiles when the sole remaining virtual numeraire parameter takes the value $0\le a \le 1$.\footnote{By \autoref{compcorr}, $\Delta W_a(x) = \frac{1}{a} \ln\bigg[\big(ax^1 + (1-a)\big)\big(ax^2 + (1-a)\big)\big(\frac{a}{2}(x^1 + x^2) + (1-a) \big) \bigg]$.} As $\Delta W_a$ is decreasing in $a$, a policy is robustly (weakly) welfare-enhancing if and only if $ \Delta W_1(x) \ge 0$, which obtains precisely when $x^1 x^2(x^1 + x^2) \ge 2$. Likewise, a policy is welfare-diminishing under every measurement unit entertained by the planner precisely when $\Delta W_0(x) \le 0$, or, equivalently, $x^1 + x^2 \le 2$. The remaining policies are welfare-improving under some, but not all, of the criteria the planner entertains.\hfill $\blacksquare$
\end{example}

\section{Conclusion}\label{conclusions}

This paper develops a theory of measurement units for utility representations. We show that our theory is rich enough to provide an appropriate measurement unit for any choice of representation, which, in conjunction with the underlying preference, defines it uniquely up to an additive constant.

Given that virtual commodities, our utility units, can be defined in physical terms without reference to any utility or preference information, it is meaningful to ask when a single commodity can serve as a common scale unit for a profile of utilities.  The presence of such a common unit then provides a precise measurement-theoretic foundation for the types of interpersonal utility comparisons commonly encountered in social choice and welfare economics.

By aggregating comparable utilities in utilitarian fashion, one obtains a natural welfare criterion, which trades off the comparable gains and losses of different individuals against each other equally. Even without a unique comparable profile, comparability substantially restricts the possible social welfare evaluations, in many instances delivering unique rankings, even between non-Pareto ranked alternatives. We illustrate how, in the case of expected utility preferences, our theory naturally extends and unifies each of the most well-known approaches to aggregation, and delivers clear and incisive insights in practical problems of welfare analysis.

\bibliographystyle{ecta}
\bibliography{bibliography}

\begin{bibunit}
\begin{appendix}

    \section{A Flow Lemma}

    Virtual commodities, as defined in \autoref{commsect}, simply correspond to maximal flows on $\mathcal{X}$ (see Chapter 9 of  \citealt{Lee2002} for formal definitions). This gives rise to a one-to-one correspondence between virtual commodities $(\mathcal{D}, \varphi)$, and their generators, $\dot{\varphi}$.

    \begin{lemma}\label{flowlemma}
        Every virtual commodity admits a unique generator, and conversely, any vector field on $\mathcal{X}$ is the generator of a unique virtual commodity. Moreover, every virtual commodity $(\mathcal{D}, \varphi)$ satisfies:
        \begin{itemize}
            \item[(i)] For any $x \in \mathcal{X}$, if $t \in \mathcal{D}^{(x)}$, then:
            \[
                \mathcal{D}^{\big(\varphi_t(x)\big)} = \mathcal{D}^{(x)} - t = \big\{t'-t : t' \in \mathcal{D}^{(x)} \big\}.
            \]
            \item[(ii)] For all $t \in \mathbb{R}$, the set $\mathcal{X}_t = \{x \in \mathcal{X} : (x,t) \in \mathcal{D}\}$ is open in $\mathcal{X}$, and the map $\varphi_t$ is a diffeomorphism $\mathcal{X}_t \to \mathcal{X}_{-t}$, with inverse $\varphi_{-t}$.
        \end{itemize}
    \end{lemma}
    \begin{proof}
        See \cite{Lee2002}, Theorem 9.12.
    \end{proof}
    
    \section{Proof of \autoref{virtnumthm}}

    \begin{proof}
        We first prove claim (i). Suppose $\succsim$ is a preference, and $(\mathcal{D}, \varphi)$ is a virtual numeraire for it. We will show there exists some representation $v$ of $\succsim$ which satisfies \eqref{uu}. Since $\succsim$ is a preference, it admits some smooth, critical-point-free utility representation $u: \mathcal{X} \to \mathbb{R}$.
        
        Suppose $x \sim y$. By (\hyperlink{n1}{N.1}), for any $t \in \mathcal{D}^{(x)} \cap \mathcal{D}^{(y)}$, we have $u\big(\varphi_t(x)\big) = u\big(\varphi_t(y)\big)$. Differentiating both sides and evaluating at $t=0$ yields 
        $ \dot{\varphi}(x) \cdot \nabla u(x)$ = $ \dot{\varphi}(y) \cdot \nabla u(y) > 0$.\footnote{Positivity follows by noting that if, e.g., $\dot{\varphi}(x) \cdot \nabla u(x) = 0$, then (\hyperlink{n1}{N.1}) would imply the same of all $x'$ on the same indifference surface near $x$. Thus, for any such $x'$, $\dot{\varphi}(x')$ is tangent to the $x$-indifference surface, thus the trajectory of $\dot{\varphi}$ starting from $x$ (i.e.\ $\varphi_t(x)$ for small $t$) must belong to this indifference surface, violating (\hyperlink{n2}{N.2}).} Since $x$ and $y$ were arbitrary points belonging to a general level set, there exists some smooth function $h: \textrm{range}(u) \to \mathbb{R}_{++}$ such that $\dot{\varphi}(x) \cdot \nabla u(x) = (h\circ u)(x)$. Fix any $\bar{q} \in \textrm{range}(u)$. Then, defining $H: \textrm{range}(u) \to \mathbb{R}$ via:
\[
    H(q) = \int_{\bar{q}}^q \frac{1}{h(\xi)} \, d\xi
\]
yields a strictly increasing, critical-point-free map.  But by the chain rule:
\[
    \dot{\varphi}(x) \cdot \nabla (H \circ u)(x) = H'(u(x))  \dot{\varphi}(x) \cdot \nabla u(x) = \frac{1}{(h \circ u)(x)} \dot{\varphi}(x) \cdot \nabla u(x) = 1.
\]

Thus $v = (H \circ u)$ represents $\succsim$ and has constant $X$-directional derivative equal to unity. Therefore:
\[
    \begin{aligned}
        v\big(\varphi_t(x)\big) & = v(x) + \int_0^t (\dot{\varphi} \cdot \nabla v) \big\vert_{\varphi_\xi(x)} \, d\xi = v(x) + \int_0^t 1 \, d\xi = v(x) + t,
    \end{aligned}
\]
and hence $v$ satisfies \eqref{uu} as desired.

We now prove (ii). Toward this end, let $u$ be an arbitrary utility representation of $\succsim$. It suffices to exhibit a virtual commodity $(\mathcal{D}, \varphi)$ such that $u$ satisfies \eqref{uu}. Consider the vector field $F(x) = \nabla u(x) / \big\| \nabla u(x)\big\|_2^2$. Let $\varphi$ denote the unique, smooth, maximal flow associated with the vector field $F$, with domain $\mathcal{D}$. By \autoref{flowlemma}, this defines a virtual commodity. But now, by the fundamental theorem of calculus for line integrals:
\[
    \begin{aligned}
        u\big(\varphi_t(x)\big) - u(x) = \int_0^t  (F \cdot \nabla u) \big\vert_{\varphi_\xi(x)} \, d\xi = \int_0^t 1 \, d\xi = t.
    \end{aligned}
\]
Thus $(\mathcal{D}, \varphi)$ and $u$ satisfy \eqref{uu} and hence $(\mathcal{D}, \varphi)$ is a virtual numeraire for $\succsim$.

Finally, suppose $u$ and $v$ are representations of $\succsim$, and both satisfy \eqref{uu} for some common virtual numeraire $(\mathcal{D}, \varphi)$. Then $\dot{\varphi} \cdot \nabla u = 1 = \dot{\varphi} \cdot \nabla v$ pointwise. By the chain rule, $\nabla u$ and $\nabla v$ are additionally everywhere collinear, implying they must be equal. As a consequence, $u$ and $v$ are equal up to a constant of integration.
\end{proof}

\section{Proof of \autoref{utiliarianthm}}

We will rely on the following technical lemma.

\begin{lemma}\label{richnesslemma}
    For any distinct $x_1,\ldots, x_K \in \mathcal{X}$, and any $A \in \mathbb{R}^{K \times S}$, there exists a profile $u \in \mathcal{C}$ such that:
    \[
        a_{ks} = u^s(x_k)
    \]
    for all $s = 1,\ldots, S$ and $k = 1,\ldots, K$.
\end{lemma}
\begin{proof}
    Let $V = \textrm{span}\{x - x' : x, x' \in \mathcal{X}\}$, and fix $\{x_1,\ldots, x_K\} \subset \mathcal{X}$. Since $\dim{V} = \dim{\mathcal{X}} = d > 1$, there exists a vector $v \in V$ such that $v$ is not parallel to $x_k - x_{k'}$, for all distinct $k,k' = 1,\ldots , K$, and let $e_1, \ldots, e_{d-1}, v$ denote any extension of $v$ into a basis for $V$.

    Now, let $x_0 \in \mathcal{X}$ be an arbitrary, fixed allocation. Then every $x \in \mathcal{X}$ can be written uniquely as:
    \[
        x = x_0 + \sum_{i = 1}^{d-1} y^i e_i + tv,
    \]
    thus we may express each $x$ uniquely in $(y,t)$ coordinates. Writing $x_k = (y_k, t_k) = (y^1_k, \ldots, y^{d-1}_k, t_k)$, note that for all $k \neq k'$, we have $y_k \neq y_{k'}$, as if this were not true, $x_k - x_{k'} = (t_k - t_{k'})v$, a contradiction.

    Define, for each $k =1,\ldots, K$:
    \[
        \lambda_k(y) = \prod_{\substack{k'=1 \\ k' \neq k}}^K \frac{\|y- y_{k'} \|^2}{\|y_k - y_{k'}\|^2}.
    \]
    By the preceding argument, this is well-defined as its denominator is always non-zero. Moreover, for any $j,k = 1,\ldots, K$, $\lambda_k(y_j)$ is one if $j = k$ and zero otherwise. Finally, let:
    \[
        u^s(y,t) = t +  \sum_{k=1}^K (a_{ks} - t_k) \lambda_k(y).
    \]
    Then at $x_{k'}$, we have:
    \[
        \begin{aligned}
            u^s(y_{k'}, t_{k'}) & = t_{k'} + \sum_{k=1}^K (a_{ks} - t_k) \lambda_k(y_{k'})\\
            & = t_{k'} + (a_{k's} - t_{k'}) \\
            & = a_{k's}
        \end{aligned}
    \]
    as desired. As each $u^s$ is clearly $t$ quasilinear, $\varphi_t(x) = x + tv$ is a common utility unit, and thus $u = (u^1, \ldots, u^S) \in \mathcal{C}$.
\end{proof}

\subsection{Proof of \autoref{utiliarianthm}}

\begin{proof}
    In light of \autoref{richnesslemma}, our proof essentially mirrors that of Theorem 3 of \cite{dAspremont1977}. We describe the changes that are needed below.

    Lemma 2 of \cite{dAspremont1977} considers only profiles whose utility values are specified at three alternatives. By \autoref{richnesslemma} such profiles always exist in $\mathcal{C}$ hence their Lemma 2, in its original proof, goes through essentially unchanged. Thus our $F$ satisfies their axiom (XNE).

    Now, for any $a,b \in \mathbb{R}^S$ define $a \, R^* \, b$ if and only if there exist $x, x' \in \mathcal{X}$ and $u \in \mathcal{C}$ where $u(x) = a$, $u(x') = b$, and $x \succeq_u x'$. By \autoref{richnesslemma} such $x,x'$ and $u$ can always be found. Then Lemma 3 of \cite{dAspremont1977} goes through unchanged for this relation $R^*$.

    The proof of Lemma 4 in \cite{dAspremont1977} is the sole other point in which a slightly modified construction is needed. Writing $I^*$ as the symmetric component of $R^*$, we wish to show that for any permutation $\sigma: \mathcal{S} \to \mathcal{S}$, and any $a \in \mathbb{R}^S$, $a I^* \sigma a$, where $\sigma a = \big(a^{\sigma(1)}, \ldots, a^{\sigma(S)}\big)$. It suffices to consider $\sigma$ a transposition, as every permutation is a product of transpositions. Thus, let $b = \sigma a$ and by \autoref{richnesslemma}, suppose $u\in \mathcal{C}$ and $x,x'\in \mathcal{X}$ satisfies $u(x) = a$ and $u(x') = b$. Let $\sigma u = (u^{\sigma(1)}, \ldots, u^{\sigma(S)})$. By (\hyperlink{a}{A}), $\succeq_u \, = \, \succeq_{\sigma u}$. Since $\sigma^2 = \textrm{id}$, we have $\sigma u(x) = b$ and $\sigma u (x') = a$, hence:
    \[
        a R^* b \quad \iff \quad x \succeq_u x' \quad \iff \quad x \succeq_{\sigma u} x' \quad \iff \quad b R^* a.
    \]
    The claim for arbitrary permutations $\sigma$ then follows. Lemma 5 of \cite{dAspremont1977} then goes through unchanged, and likewise the remainder of the proof of their Theorem 3.
\end{proof}

\section{Proof of \autoref{cardexistance}}
\begin{proof}
First, note that if at some $x \in \mathcal{X}$, $0 \not \in \textrm{Aff}\big\{\nabla u^1(x), \ldots, \nabla u^S(x)\big\}$, then:
\begin{equation}\label{cardconsist2}
    \sum_{s\in \mathcal{S}} c^s \nabla u^s(x) = 0 \quad \implies \quad \sum_{s \in \mathcal{S}} c^s = 0,
\end{equation}
as otherwise the vector $c$ could be rescaled to sum to unity. Conversely, if $0 \in \textrm{Aff}\big\{\nabla u^1(x), \ldots, \nabla u^S(x)\big\}$, there exists $c \in \mathbb{R}^\mathcal{S}$, such that:
\[
     \sum_{s\in \mathcal{S}} c^s \nabla u^s(x) = 0 \quad \textrm{ and } \quad \sum_{s \in \mathcal{S}} c^s = 1.
\]
Thus \eqref{cardconsist2} is equivalent to \eqref{cardconsist}.

Now, let $r = r(\mathcal{S})$. As the profile has constant rank, at every $x \in \mathcal{X}$, there exist $\nabla u^{s_1},$ $\ldots,$ $\nabla u^{s_{r}}$ which are linearly independent on some neighborhood of $x \in \mathcal{X}$. Thus, for every $x \in \mathcal{X}$, there exists an open, connected neighborhood $U_x$ of $x$, and selection of $r$ gradients $\nabla u^{s_1}, \ldots, \nabla u^{s_r}$, such that these gradients remain linearly independent on $U_x$. For each $x \in \mathcal{X}$, let $\mathcal{A}_x$ denote the $d \times r$ matrix-valued function whose columns consist of these gradients.

Then, for every $(U_x, \mathcal{A}_x)$, and every $x' \in U_x$, the $r \times r$ matrix $\mathcal{A}_x^\intercal(x') \mathcal{A}_x(x')$ is invertible, hence we may define a function $F_{x}$ on $U_x$ via:
\[
    F_x(x')= \mathcal{A}_x(x')\big[\mathcal{A}_x^\intercal \mathcal{A}_x(x') \big]^{-1} \mathbbm{1},
\]
where $\mathbbm{1}$ denotes the column vector of $r$ ones. This function is smooth and, by construction, $\mathcal{A}_x(x')^\intercal F_x(x') = \mathbbm{1}$. Thus $F_x(x')$ has dot product one with every gradient $\nabla u^{s_1}(x'),\ldots, \nabla u^{s_r}(x')$, at every $x' \in U_x$.  Now, on $U_x$, every remaining gradient can be written as:
\[
    \nabla u^s(x') = \sum_{i = 1}^r c^i(x') \nabla u^{s_i}(x'),
\]
where the coefficient functions $c^i$ are smooth. By \eqref{cardconsist2}, the weight functions must sum to unity, as $\nabla u^s(x') - \sum_{i = 1}^r c^i(x') \nabla u^{s_i}(x') = 0$. Thus $ F_x(x') \cdot \nabla u^s(x') = \sum_{i=1}^r c^i(x') F_x(x')\cdot \nabla u^{s_i}(x') = \sum_{i =1}^r c^i(x') = 1$. Thus, each $F_x$ has dot product equal to one with every gradient $\nabla u^1, \ldots, \nabla u^S$ at every point of $U_x$.

Since $\mathcal{X}$ is a subset of Euclidean space it is paracompact, therefore the open cover $\mathcal{U} = \{U_x\}_{x \in \mathcal{X}}$ admits a locally finite refinement, $\tilde{\mathcal{U}} = \{\tilde{U}_i\}_{i \in I}$. For each $i \in I$, select any $x_i$ such that $\tilde{U}_i \subseteq U_{x_i}$. Now, let $\{\psi_i\}_{i \in I}$ denote a smooth partition of unity subordinate to $\tilde{\mathcal{U}}$. Then:
\[
    F(x) = \sum_{i \in I} \psi_i(x) F_{x_i}(x)
\]
is smooth and well-defined, as the open cover $\{\tilde{U}_i\}_{i \in I}$ is locally finite, hence for any $x \in \mathcal{X}$, the sum is taken over only finitely many non-zero terms. Moreover, as at every point it is a convex combination of vectors which have dot product equal to one with respect to every gradient, so too does $F$. Then, by an identical argument to that appearing in the proof of \autoref{virtnumthm}, there is then a unique virtual commodity $(\mathcal{D}, \varphi)$ on $\mathcal{X}$ such that $\dot{\varphi} = F$ pointwise. This $\varphi$ then necessarily satisfies \eqref{uu} for every $u^s$.

Conversely, suppose there exists some common utility unit $(\mathcal{D}, \varphi)$ for every $u^s$. Thus, for every $u^s$, we have $\dot{\varphi} \cdot \nabla u^s = 1$. For sake of contradiction, suppose, now that $0 \in \textrm{Aff}\big\{\nabla u^1(x), \ldots, \nabla u^S(x)\big\}$.  By \eqref{cardconsist2} there must be some $c \in \mathbb{R}^\mathcal{S}$ such that $\sum_s c^s \nabla u^s(x)= 0$ and $\sum_s c^s = 1$. But this means $0 = \dot{\varphi} \cdot \sum_s c^s \nabla u^s = \sum_s c^s  \dot{\varphi}\cdot \nabla u^s = \sum_s c^s = 1$, a contradiction. Thus 
$0 \not \in \textrm{Aff}\big\{\nabla u^1(x), \ldots, \nabla u^S(x)\big\}$, as desired.
\end{proof}

\section{Proof of \autoref{locseptheorem}}

\begin{proof}
    Suppose $\mathcal{N} \subseteq \mathcal{X}$ is a connected neighborhood of $x_0$ on which \eqref{seprepglobal} holds, for utilities $v^1, \ldots, v^{d+1}$ and some vector $\lambda \in \{-1,1\}^d$. Define $v(x) = \big(v^1(x), \ldots, v^d(x)\big)$, and let $\mathcal{V} = v(\mathcal{N})$. Let $\varphi$ be a virtual commodity on $\mathcal{V}$ such that it is a common numeraire for the preferences $\succsim^1, \ldots, \succsim^d$ represented by the coordinate functions $v^i$, $i =1,\ldots, d$ and $\succsim^{d+1}$ represented by $\lambda^1 v^1 + \cdots + \lambda^d v^d$. By (\hyperlink{n1}{N.1}) applied to each $\succsim^i$, $i = 1,\ldots, d$, we have:
    \begin{equation} \label{separableode}
        \dot{\varphi}^i(v) = F^i(v),
    \end{equation}
    for some tuple of functions $(F^1, \ldots, F^d)$, defined on $\mathcal{V}$ and satisfying $F^i_j = 0$ for all $i \neq j$. But by (\hyperlink{n1}{N.1}) for $\succsim^{d+1}$, we obtain:
    \[
        F^{d+1}(v) = \sum_{i=1}^d \lambda^i F^i(v^i)
    \]
    is constant on subsets of vectors with constant $\lambda$ dot product in $\mathcal{V}$. Let $\xi^i$ denote the $d$-vector of all $0$'s except for its $i$th component, which is $1$. Then, for all $i \neq j$:
    \[
        DF^{d+1}(v)(\lambda^i \xi^i - \lambda^j \xi^j) = 
         F^i_i(v) -  F^j_j(v) = 0.
    \]
    Since $v$ is a diffeomorphism on $\mathcal{N}$, $\mathcal{V}$ is open. Thus each coordinate around any fixed $v(x)$ can be independently varied, hence the above implies every $F^i_i$ is locally constant and $F^i_i = F^j_j$ for all $i,j$. Since $\mathcal{N}$ is connected, we obtain that each $F^i$ is the restriction to $\mathcal{V}$ of a function:
    \[
       F^i(v) = a v^i + b^i.
    \]
    By (\hyperlink{n2}{N.2}), for each $i = 1,\ldots, d$, we have $a v^i + b^i > 0$, as well as $\lambda \cdot \big(av^i + b\big) > 0$. Conversely, clearly any $\varphi$ satisfying \eqref{separableode} also satisfies (\hyperlink{n1}{N.1}) and (\hyperlink{n2}{N.2}). Thus every virtual numeraire takes this form. The claim then follows from solving the ordinary differential equation \eqref{separableode} and conjugating back to $\mathcal{N}$ via $v$. 
\end{proof}

\section{Proposition Proofs}

\subsection{Proof of \autoref{decompprop}}

    \begin{proof}
        Fix a choice of utility representation $u:\mathcal{X} \to \mathbb{R}$ for $\succsim$, and suppose $y \succsim x$ for some $x,y \in \mathcal{X}$. First, note if $x \sim y$, then the singleton sequence $(x_1, t_1) = (x,0)$ trivially belongs to $\mathcal{D}$ and suffices for the claim. Thus suppose $y \succ x$, and consider the interval $A = \big[ u(x), u(y) \big] \subset \mathbb{R}$. By convexity of $\mathcal{X}$, for every $a \in A$ there exists some $z_a \in \mathcal{X}$ with $u(z_a) = a$. Now, for each $a \in A$, let:
        \[
            I_a = u\big(\varphi(z_a,\mathcal{D}^{(z_a)})\big)
        \]
        for an arbitrary choice of $z_a \in u^{-1}(a)$. Since $u$ is continuous, and by (\hyperlink{n2}{N.2}), $u\big(\varphi(z_a,t)\big)$ is strictly increasing in $t$, the fact $\mathcal{D}^{(z_a)}$ is an open interval containing zero implies $I_a \subseteq \mathbb{R}$ is an open interval containing $a$.

        Thus $\{I_a\}_{a \in A}$ defines an open cover of the compact interval $A$, and hence there exists a finite $A' \subset A$ such that $\{I_a\}_{a \in A'}$ forms a finite subcover; by the  Lebesgue number lemma, there also exists a partition:
        \[
            u(x) = b_1 < b_2 < \cdots < b_K < b_{K+1} = u(y),
        \]
        such that every closed interval $[b_k, b_{k+1}]$, $k = 1, \ldots, K$ is contained in some element of $\{I_a\}_{a \in A'}$. Now, for every $k = 1,\ldots, K$, let $I_{a_k}$, $a_k \in A'$ be any sub-cover element such that $[b_k, b_{k+1}] \subset I_{a_k}$.
        Thus there exist $t^-_k < t^+_k$ such that:
        \[
            \varphi_{t^-_k}(z_{a_k}) \sim z_{b_k} \quad \textrm{ and } \quad \varphi_{t^+_k}(z_{a_k}) \sim z_{b_{k+1}},
        \]
        and hence:
        \[
            \varphi_{t_k^+ - t_k^-}\big(\varphi_{t_k^-}(z_{a_k})\big) \sim z_{b_{k+1}}.
        \]
        In particular, by \autoref{flowlemma}, the left-hand term is well-defined. Then, defining:
        \[
            x_k = \varphi_{t_k^-}(z_{a_k}), \quad \textrm{ and }\quad  t_k = t_k^+ - t_k^-,
        \]
        we have:
        \[
            \varphi_{t_k}(x_k) = \varphi_{t_k^+ - t_k^-}\big(\varphi_{t_k^-}(z_{a_k})\big) \sim z_{b_{k+1}} \sim \varphi_{t_{k+1}^-}(z_{a_{k+1}}) = x_{k+1},
        \]
        as desired. Now, by maximality of $(\mathcal{D}, \varphi)$, if $(z,t) \in \mathcal{D}$, then $\big(\varphi_t(z), -t\big) \in \mathcal{D}$ by \autoref{flowlemma}. Thus, if $x \succ y$, applying the same construction in the opposite direction, and then reversing the transfers, yields the analogous conclusion. Finally, when $(\mathcal{D}, \varphi)$ is a utility unit for $u$, the last claim follows from \eqref{uu}, as the sum along the chain telescopes.
    \end{proof}

\subsection{Proof of \autoref{uniquenessprop}}

\begin{proof}
    For sufficiency, it suffices to note that if $(\mathcal{D}, \varphi)$ is a common unit for $(u^1,\ldots, u^S)$, then $(\tilde{\mathcal{D}}, \tilde{\varphi})$, where $\tilde{\varphi}_t = \varphi_{t/\alpha}$, and $\tilde{\mathcal{D}}^{(x)} = \alpha \mathcal{D}^{(x)}$, is a unit for $(\alpha u^1 + \beta^1, \ldots, \alpha u^S + \beta^S)$.
    
    For necessity, suppose $S \ge 2$, and that for some tuple of increasing, smooth, and critical-point-free functions $f \equiv (f^1,\ldots, f^S)$, whenever a profile $(u^1,\ldots, u^S)$ is interpersonally comparable, so is $\big(f^1 \circ u^1, \ldots, f^S \circ u^S\big)$.

    Choose global coordinates on $\mathcal{X}$ such that the projection of $\mathcal{X}$ onto the first factor is non-singleton.  Define $u^s(x) = x^1 + c^s$ for all $s = 1, \ldots, S$, where each $c_s \in \mathbb{R}$. Then $(u^1, \ldots, u^S)$ are interpersonally comparable, with common utility unit given by $(\mathcal{D}, \varphi)$, where $\varphi_t(x) = x + (t,0,\ldots, 0)$ and $\mathcal{D}$ is the maximal domain associated with this flow on $\mathcal{X}$.\footnote{This follows from \autoref{flowlemma}.}

    Now, if $(f^1 \circ u^1, \ldots, f^S \circ u^S)$ are interpersonally comparable, there exists some $(\tilde{\mathcal{D}}, \tilde{\varphi})$ such that $(f^s \circ u^s)\big(\tilde{\varphi}_t(x)\big) = (f^s \circ u^s)(x) + t.$ Differentiating both sides with respect to $t$ and evaluating at $t=0$ implies:
    \[
        \dot{\tilde{\varphi}} \cdot \nabla u^s = \frac{1}{(f^s)'(u^s)}
    \]
    for every $s =1 ,\ldots, S$. In particular, since $\nabla u^s = [1, 0,\ldots, 0]^\intercal$, this implies:
    \[
        \dot{\tilde{\varphi}}^1(x) = \frac{1}{(f^s)'(x^1 + c^s)}.
    \]
    But since $f$ preserves interpersonal comparability for any profile, it does so for the profile obtained by simply varying the $c^s$. Since any profile obtained in this fashion remains comparable, $(f^1)' = \cdots = (f^S)' = \alpha$ for some constant $\alpha$; since each $f^s$ is strictly increasing, $\alpha > 0$ and $f^s(t) = \alpha t + \beta_s$.
\end{proof}

\subsection{Proof of \autoref{eucorr}}
\begin{proof}
    Conditions (i) and (ii) ensure a locally separable representation \eqref{seprepglobal} holds globally, where each $v^i$, $i = 1,\ldots, d+1$ is an expected utility function. Without loss of generality, we may normalize each so that $v^i(\delta_{z_*}) = 0$. \autoref{locseptheorem} then characterizes every common numeraire; by direct verification, the case of $a= 0$ corresponds to a translation and $b=0$ to a dilation. The claim follows.
\end{proof}

\pagebreak

\section*{\huge Supplemental Appendix}\hypertarget{suppapp}{}

\section{Quasilinear Representations}\label{quasilinearsect}

Our next result illustrates that quasilinearity provides a kind of normal-form representation for interpersonally comparable profiles. Under mild regularity hypotheses, we show a profile admits a common measurement unit if, and only if, there exists a globally-defined, smooth change of coordinates that simultaneously renders each utility quasilinear.\footnote{By smooth change of coordinates, we mean a smooth bijection with smooth inverse.} This illustrates that the abstract machinery of virtual commodities may simply be viewed as a coordinate-free generalization of classical commodities, allowing us to speak of numeraire goods even outside the setting of traditional consumer theory.

Suppose $(\mathcal{D}, \varphi)$ is a virtual commodity, and $x, x' \in \mathcal{X}$.  Let $x \sim_\varphi x'$ denote that $x$ and $x'$ are equal, up to some quantity of additional commodity, that is, $x' = \varphi_t(x)$ for some $t \in \mathcal{D}^{(x)}$ or vice versa. We say $(\mathcal{D}, \varphi)$ is a {\bf regular} virtual commodity if the binary relation $\sim_{\varphi}$ has a closed graph.

\begin{theorem}\label{quasilinearizationthm}
    A regular virtual commodity $(\mathcal{D}, \varphi)$ is a common utility unit for the profile $(u^1,\ldots, u^S)$ if, and only if, there exists a smooth manifold $\mathcal{Y}$, open subset $\Omega \subseteq \mathbb{R} \times \mathcal{Y}$, and smooth change-of-coordinates $h:\mathcal{X} \to \Omega$, such that:
    \begin{itemize}
        \item[(i)] For all $(x,t) \in \mathcal{D}$,
        \[
            h\big(\varphi_t(x)\big) = h(x) + (t, 0, \ldots, 0).
        \]
        \item[(ii)] For all $s \in \mathcal{S}$ and $(t,y) \in \Omega$,
        \[
            u^s\big(h^{-1}(t,y)\big) = t + v^s(y).
        \]
    \end{itemize}
\end{theorem}

\subsection{Proof of \autoref{quasilinearizationthm}}

\begin{proof}
    $(\Longrightarrow)$: Suppose that $(\mathcal{D}, \varphi)$ is a virtual commodity which is a common utility unit for a profile $(u^1,\ldots, u^S)$, and that $\sim_\varphi$ is a closed binary relation. First, we note $\sim_\varphi$ is an equivalence relation. It is defined symmetrically, and since for any $x \in \mathcal{X}$, $0 \in \mathcal{D}^{(x)}$ and $\varphi_0(x) = x$, it is reflexive. Suppose then that $x\sim_\varphi y$ and $y \sim_\varphi z$. If $\varphi_t(x) = y$ and $\varphi_{t'}(y) = z$, then $\varphi_{t'+t}(x)$ is well-defined and equals $z$. Thus suppose $\varphi_t(x) = y$ and $\varphi_{t'}(z) = y$. Then $\varphi_{-t}(y) = x$ and hence $\varphi_{t'-t}(z)$ is well-defined and equals $x$. Thus $\sim_{\varphi}$ is transitive and hence an equivalence relation.

    We now show the quotient map $\pi: \mathcal{X} \to \mathcal{X}/\!\sim_\varphi$ is open. First, note that for any $t \in \mathbb{R}$, let $\mathcal{D}_t = \{x \in \mathcal{X} : t \in \mathcal{D}^{(x)}\}$. As $\mathcal{D}$ is open, so too is each $\mathcal{D}_t$. Moreover, $\varphi_t : \mathcal{D}_t \to \mathcal{D}_{-t}$ is a diffeomorphism with inverse $\varphi_{-t}$. Now, let $U \subseteq \mathcal{X}$ be open. Then:
    \[
        \begin{aligned}
            \pi^{-1}\big(\pi(U)\big) & = \big\{x \in \mathcal{X} : \varphi_t(x) \in U \textrm{ for some } t \in \mathbb{R}\big\}\\
            & = \bigcup_{t \in \mathbb{R}} \varphi_t\big(U \cap \mathcal{D}_t\big),
        \end{aligned}
    \] 
    which is open. Thus, in the quotient topology on $\mathcal{X}/\!\sim_\varphi$, $\pi(U)$ is open for any open $U \subseteq \mathcal{X}$, and hence $\pi$ is an open map.

    We now show that $\mathcal{X}/\!\!\sim_\varphi$ is a smooth manifold. First, we verify the quotient is Hausdorff. Let $\pi(x) \neq \pi(x')$. Then $(x,x') \not \in \, \sim_\varphi$ and hence there exist open neighborhoods $U$ of $x$ and $V$ of $x'$ in $\mathcal{X}$ such that $(U \times V)\, \cap \, \sim_{\varphi} = \varnothing$. As both $U$ and $V$ are open, so too are $\pi(U)$ and $\pi(V)$, as $\pi$ is an open map. Then $\pi(U) \cap \pi(V) = \varnothing$, as their product is disjoint from the graph of $\sim_\varphi$, and thus $\mathcal{X}/\!\sim_\varphi$ is Hausdorff. Moreover, since $\mathcal{X}$ is second-countable, so too is $\mathcal{X}/\!\sim_\varphi$, as by openness of $\pi$, the image of any countable basis for $\mathcal{X}$ under $\pi$ forms a countable basis for $\mathcal{X}/\!\sim_\varphi$.

    We now construct a smooth atlas for $\mathcal{X}/\!\sim_\varphi$. Fix any $s \in \mathcal{S}$, and let $x \in \mathcal{X}$. Let $\mathcal{I}$ denote the $u^s$-indifference surface through $x$. Because $u^s$ is critical-point-free, $\mathcal{I}$ is an embedded smooth submanifold of $\mathcal{X}$, e.g.\ \cite{Lee2002} Corollary 5.14. Now, let $U \subset \mathcal{I}$ be a coordinate neighborhood of $x$ on this submanifold. Define:
    \[
        \mathcal{D}_U = \big\{(x',t') \in U \times \mathbb{R}: (x',t') \in \mathcal{D}\big\},
    \]
    and let $\Psi: \mathcal{D}_U \to \mathcal{X}$ via $\Psi(x',t') = \varphi_{t'}(x')$. Suppose $\Psi(x',t') = \Psi(x'',t'')$, i.e.\ $\varphi_{t'}(x') = \varphi_{t''}(x'')$. Then $u^s(x') = u^s(x'')$, hence by \eqref{uu}, $u\big(\varphi_{t'}(x')\big) - u^s\big( \varphi_{t''}(x'')\big) = t' - t''$ and thus $t' = t''$. But as $\varphi_{t'}: \mathcal{D}_{t'} \to \mathcal{D}_{-t'}$ is injective, we obtain $x' = x''$. Thus $\Psi$ is itself injective.
      
    Define $X(x) = \frac{\partial}{\partial t} \varphi_t(x) \big\vert_{t=0}$. Fix some $(x',t') \in \mathcal{D}_U$, and let $x'' = \varphi_{t'}(x')$. Since $U$ is open in $\mathcal{I}$, a vector  $\xi$ is tangent to $\mathcal{U}$ at $x'$ if and only if $\xi \in \textrm{ker}\, du_{x'}$. Then for any $\xi \in T_{x'} U$ and $\lambda \in \mathbb{R}$, 
    \begin{equation} \label{psidiff}
        d \Psi_{(x',t')}(\xi, \lambda) = d(\varphi_{t'})_{x'}(\xi) + \lambda X_{x''}.
    \end{equation}
    These two terms belong to complementary subspaces, as $\langle X, \nabla u\rangle = 1$ implies $X_{x''} \not \in \textrm{ker} \, du_{x''}$ and hence is not tangent to $U$. Now, differentiating both sides of \eqref{uu} with respect to $x$ and evaluating at $(x',t')$ yields:
    \[
        du_{x''} \circ d(\varphi_{t'})_{x'}  = du_{x'},
    \]
    and hence $d(\varphi_{t'})_{x'}(T_{x'} U) \subseteq \textrm{ker} \, du_{x''}$; as $d(\varphi_{t'})_{x'}$ is an isomorphism, these two subspaces have equal dimension, and hence coincide. Now, suppose:
    \[
        d \Psi_{(x',t')}(\xi, \lambda) = 0.
    \]
    Applying $du_{x''}$ to both sides yields:
    \[
        0 = du_{x''}\bigg(d(\varphi_{t'})_{x'}(\xi) + \lambda X_{x''}\bigg) = \lambda
    \]
    and hence $\lambda = 0$. Plugging this into \eqref{psidiff}, this implies $0 = d(\varphi_{t'})_{x'}(\xi)$. But as noted above, $d(\varphi_{t'})_{x'}$ is an isomorphism hence $\xi = 0$. Thus $d\Psi_{(x',t')}$ is an isomorphism and therefore $\Psi$ is a local diffeomorphism.

    Now, by definition, the range $\Psi(\mathcal{D}_U) = \pi^{-1}\big(\pi(U)\big)$, hence $\pi^{-1}\big(\pi(U)\big)$ is open, and therefore $\Psi$ is a diffeomorphism. But this implies that $\pi(U)$ is an open subset of $X/\!\sim_\varphi$, and that $\pi \big\vert_U \to \pi(U)$ is a homeomorphism onto an open subset. As $U$ is a coordinate neighborhood, this yields charts covering $X/\!\sim_{\varphi}$.

    Suppose $U, V \subset X/\!\sim_{\varphi}$ are two coordinate neighborhoods such that $U \cap V \neq \varnothing$.  Then for some open subset $U' \subseteq \{x \in \mathcal{X} : u^s(x) = c\}$ and $V' \subseteq \{x \in \mathcal{X} : u^s(x) = d\}$, the transition map is simply the restriction of $\varphi_{d-c}$, which is smooth. Hence these charts define a smooth $(d-1)$-manifold structure on $X/\!\sim_{\varphi}$.

    We now are in position to define $h$, $\mathcal{Y}$, and $\Omega$. Let $\mathcal{Y} = X/\!\sim_\varphi$, and define $h: \mathcal{X} \to  \mathbb{R} \times \mathcal{Y}$ via:
    \[
        h(x) = \big(u^s(x), \pi(x)\big).
    \]
    The map $h$ is injective, as $h(x) = h(x')$ implies $\pi(x) = \pi(x')$ and hence $\varphi_t(x) = x'$. But since $u^s(x) = u^s(x')$, by \eqref{uu}, we must have $x= x'$. Suppose now $dh_x(\zeta) = 0$ for some tangent vector $\zeta \in T_x \mathcal{X}$. Since $dh_x = \big[du^s_x \;\; d\pi_x\big]^\intercal$, this implies $d\pi_x(\zeta) = 0$ hence $\zeta$ is proportional to $X_x$. But since $du^s_x(\lambda X_x) = \lambda$, it follows $\lambda = 0$, and hence $\zeta = 0$. Thus $dh_x$ is injective and since both the domain and codomain are $d$-dimensional, $dh_x$ is an isomorphism, making $h$ a local diffeomorphism. Because we have already shown $h$ to be injective, and both $\mathcal{X}$ and $\mathbb{R} \times X/\!\sim_\varphi$ are $d$-dimensional manifolds, $\mathcal{X}$ without boundary, we obtain that $h$ is a diffeomorphism from $\mathcal{X}$ onto the open set $h(\mathcal{X}) \subseteq \mathbb{R} \times X/\!\sim_\varphi$ (e.g.\ Prop.\ 4.22.(d) in \citealt{Lee2002}). We then define $\Omega = h(\mathcal{X}) \subseteq\mathbb{R} \times X/\!\sim_\varphi$. 

    Finally, we now verify (i) and (ii).  For (i), note that for all $(x,t)\in \mathcal{D}$,
    \[
        \begin{aligned}
            h\big(\varphi_t(x)\big) & = \big(u^s(\varphi_t(x))\big), \pi(\varphi_t(x)\big)\\
            & = \big(u^s(x) + t, \pi(x)\big)\\
            & = h(x) + (t, 0,\ldots, 0),
        \end{aligned}
    \]
    where the second equality follows from the fact $u^s$ satisfies \eqref{uu}. For (ii), note that from the above, $(h_* X)_{(t,y)} = dh_{h^{-1}(t,y)}(X_{h^{-1}(t,y)}) = \partial_t$. Thus if $h(x) = (t,y)$, then for any $s' \in \mathcal{S}$:
    \[
        \frac{\partial}{\partial t} u^{s'}\big(h^{-1}(t,y)\big) = \langle X, \nabla u\rangle \big \vert_x = 1.
    \]
    Writing $\Omega = h(\mathcal{X})$, let $\Omega^{(y)} = \{t : (t,y) \in \Omega\}$. For any $y$, this is a connected interval, as if $h(x) = (t,y)$, this just equals $u^s\big( \pi^{-1}(\pi(x))\big) = u^s(\varphi_{\mathcal{D}^{(x)}}(x))$, which is connected as $\mathcal{D}^{(x)}$ is. As each $\Omega^{(y)}$ is connected, $u^{s'}\big(h^{-1}(t,y)\big) - t$ is then constant on $\Omega^{(y)}$ and hence depends only on $y$ as a function defined on all $\Omega$. Define this function to be $v^{s'}(y)=  u^{s'}\big(h^{-1}(t,y)\big) - t$. Then by identity, $u^{s'}\big(h^{-1}(t,y)\big) = t + v^{s'}(y)$, as desired.

    $(\Longleftarrow)$: Suppose now (i) and (ii) hold. Suppose $h(x) = (t,y)$ and let $t' \in \mathcal{D}^{(x)}$. Then:
    \[
        \begin{aligned}
            u^s(x) & = u^s\big(h^{-1}(t,y) \big)   \\
            & = u^s\big(h^{-1}(t+t', y)\big) - t'\\
            & = u^s\big(\varphi_{t'}\big(h^{-1}(t, y)\big)\big) - t'\\
            &= u^s\big(\varphi_{t'}(x)\big) - t',
        \end{aligned}
    \]
    where the first equality is an identity, the second follows from (ii), the third from (i), and the fourth again is an identity. Hence $u^s$ satisfies \eqref{uu} for each $s \in \mathcal{S}$.
    \end{proof}

\section{Locally Comparable Representations}\label{ordinalsect}

For ease of exposition, we will illustrate our results in the case of full-rank profiles, containing $S = d+1$ individuals.\footnote{All our results extend straightforwardly, albeit at the cost of more cumbersome notation, to profiles of arbitrary, finite size, as well as to profiles of lower rank; see \autoref{ordthmext}.} Consider, then, a fixed status quo allocation $x_0$, and suppose the planner observes a profile $(\succsim^1, \ldots, \succsim^{d+1})$. Let $(u^1, \ldots, u^{d+1})$ denote an arbitrary profile of representations. We will make use of the following regularity hypothesis.

\begin{itemize}
    \item[(CR.1)] \hypertarget{cr1}{} {\bf Full Profile Rank}: The subprofile $(u^1, \ldots, u^{d})$ has constant and full rank about $x_0$, i.e.\ the  matrix $\mathcal{J}(x) = \big[\nabla u^1 \cdots \nabla u^d\big]^\intercal$ is invertible in a neighborhood of $x_0$. In addition, we will suppose $\nabla u^{{d+1}}(x_0) =  \sum_{i = 1}^d c_i \nabla u^{i}(x_0)$ for positive scalars $c_1, \ldots, c_d > 0$.
\end{itemize}

Assumption (\hyperlink{cr1}{CR.1}) requires the gradients of some $d$-individual subprofile not all lie on some lower dimensional surface. We also require that the final gradient does not belong to the span of any proper sub-collection of these. However, the assumption that it is specifically a \emph{positive} linear combination is purely for notational convenience. Finally, to avoid a proliferation of indices, we will simply write $\bar{u}$ for $u^{d+1}$.

Suppose now (\hyperlink{cr1}{CR.1}) holds, and let $\mathcal{N}_0$ denote any neighborhood of $x_0$ on which $\mathcal{J}$ is invertible, and the gradient condition on $\bar{u}$ remains valid.  Given any smooth function $f: \mathcal{N}_0 \to \mathbb{R}$, we define its {\bf utility derivatives}, for any $i = 1,\ldots, d$, via:
\[
    f_i \equiv \frac{\partial f}{\partial u^{i}} = \sum_{j=1}^d \frac{\partial f}{\partial x^j} (\mathcal{J}^{-1})_{ji},
\]
These quantities are simply the derivatives of $f$, instead computed with respect to the \emph{utility} coordinates $u \equiv (u^1, \ldots, u^d)$.

This perspective allows us to consider the final utility $\bar{u}$ as a function $\bar{u}(u)$ of the utilities of the remaining members of society. Doing so, we define the {\bf separability indices} of $\bar{u}$ via:
\[
    \underbrace{\frac{\bar{u}_{ij}}{\bar{u}_i \bar{u}_j} - \frac{\bar{u}_{ik}}{\bar{u}_i\bar{u}_k}, \quad i,j,k \textrm{ distinct}}_{\textrm{If } d \, \ge \,3} \quad \textrm{ or } \quad \underbrace{\sqrt{\bigg\vert\frac{1}{\bar{u}_1\bar{u}_2} \partial_{12} \ln\bigg(\frac{\bar{u}_1}{\bar{u}_2}\bigg)\bigg\vert}}_{\textrm{If } d \, = \,2},
\]
where all subscripts refer to utility derivatives. Let $\mathbb{H}$ denote the set of such indices.\footnote{That is, when $d \ge 3$, we let $\mathbb{H}$ denote the set of all functions $\frac{\bar{u}_{ij}}{\bar{u}_i \bar{u}_j} - \frac{\bar{u}_{ik}}{\bar{u}_i\bar{u}_k}$ for distinct $i,j$, and $k$. If instead $d=2$, $\mathbb{H}$ denotes the singleton family, containing only $\sqrt{\big\vert\frac{1}{\bar{u}_1\bar{u}_2} \partial_{12} \ln\big(\frac{\bar{u}_1}{\bar{u}_2}\big)\big\vert}$.} It may be verified that, given $\bar{u}$, each $h \in \mathbb{H}$ is independent of the choices of representation $u^1,\ldots, u^d$. Moreover, while in general the value of a given $h \in \mathbb{H}$ depends on the particular choice of representation $\bar{u}$, the property that any given $h = 0$ is purely ordinal.

As we will shortly show, the indices in $\mathbb{H}$ determine when $\bar{u}$ is, up to a monotone transform, an additively separable function of the utilities $u^1, \ldots, u^d$. Equivalently, they determine when the utility-derivative marginal rates of substitution $\bar{u}_i/\bar{u}_j$ are independent of the utility levels of the rest of society $k \neq i,j$.

Before providing further interpretation, we first introduce our next regularity condition, which rules out singularities in the system $\mathbb{H}$ about $u_0 = u(x_0)$.

\begin{itemize}
    \item[(CR.2)] \hypertarget{cr2}{} {\bf Constant Index Rank}: For all $h \in \mathbb{H}$, $\textrm{sign}\big(h(u)\big)$ is constant in a neighborhood of $u_0$.\footnote{Recall $\textrm{sign}(t) = 1$ if $t > 0$, $-1$ if $t < 0$, and $0$ if $t = 0$.}
\end{itemize}

\noindent Under (\hyperlink{cr2}{CR.2}), there are two possibilities. Either every function in $\mathbb{H}$ is uniformly zero on some neighborhood of $x_0$, or there exists some $h \in \mathbb{H}$ such that $h(u_0)$ is locally strictly positive.\footnote{When $d \ge 3$, the functions corresponding to the tuple $(i,j,k)$ and $(i, k,j)$ have opposite sign. As  a consequence if any $h \in \mathbb{H}$ is non-zero at $u_0$, at least one will be strictly positive.} These two cases lead to qualitatively different behavior, hence we treat them separately.

\subsection{The Separable Case}\label{sepcase}

Suppose first every $h \in \mathbb{H}$ vanishes on some neighborhood of $u_0$. By possibly restricting further, we may suppose this is true on some $\mathcal{N} \subseteq u(\mathcal{N}_0)$.  In this case, there exist utility representations $f^1 \circ u^1,\ldots , f^d \circ u^d$ for $\succsim^1\! \!\!\,\vert_{u^{-1}(\mathcal{N})},\ldots, \succsim^d \!\!\!\,\vert_{u^{-1}(\mathcal{N})}$, and a smooth, strictly increasing, and critical-point-free $\bar{f}: \bar{u}(\mathcal{N}) \to \mathbb{R}$, such that:
\begin{equation}\label{seprep}
    (\bar{f}\circ \bar{u})(u) = \sum_{i=1}^d f^i(u^i).
\end{equation}
Let $v^i = f^i \circ u^i$, for each $i = 1,\ldots, d$. We term the tuple $v \equiv (v^1, \ldots, v^d)$ a {\bf locally separable representation} for the profile $(\succsim^1,\ldots, \succsim^{d+1})$.\footnote{The tuple of representations $v = (v^1,\ldots, v^d)$ is unique up to a component-wise transform of the form $v^i \mapsto \alpha v^i + \beta^i$, where the common $\alpha > 0$, and $\beta^i \in \mathbb{R}$.} Thus this case corresponds precisely to the locally separable setting treated in the main text.

\subsection{The Non-separable Case}

Suppose now that $\mathbb{H}$ contains some function $\bar{h}$ which is strictly positive at $x_0$. In this case, nearby $u_0 = u(x_0)$, we can define the functions:
\[
    K^i(u) = -\ln \big(\bar{h}(u)\, \bar{u}_i(u)\big),
\]
for $i = 1,\ldots, d$. In light of (\hyperlink{cr1}{CR.1}) and (\hyperlink{cr2}{CR.2}), these are well-defined. As the following proposition shows, the functions $K^1,\ldots, K^d$ generate an overdetermined system of partial differential equations, which characterize when a local virtual commodity is numeraire, for every preference in $(\succsim^1,\ldots, \succsim^{d+1})$.

\begin{proposition}\label{pdelemma}
    Assume (\hyperlink{cr1}{CR.1}) and (\hyperlink{cr2}{CR.2}), and let $\bar{h} \in \mathbb{H}$ be positive at $u_0$. Suppose $(\mathcal{D}, \varphi)$ is a local virtual commodity defined about $u_0$, and in $u$-coordinates, let:\footnote{In particular, $F^i$ is the component of $F$ pointing in direction $\nabla u^i$.}
    \[
        \dot{\varphi}(u) = F(u) = \begin{bmatrix} F^1(u)\\  \vdots \\ F^d(u) \end{bmatrix}.
    \]
    Then $(\mathcal{D}, \varphi)$ satisfies (\hyperlink{n1}{N.1}) and (\hyperlink{n2}{N.2}) for every preference in $(\succsim^1,\ldots, \succsim^{d+1})$, if, and only if, each $F^i > 0$, and $F$ satisfies:
    \begin{equation}\label{overdetpde}
    \begin{bmatrix}
        F^1_1 & \cdots & F^1_d\\
        \vdots & \ddots & \vdots\\
        F^d_1 & \cdots & F^d_d
    \end{bmatrix} = \begin{bmatrix}
         \langle F, \nabla K^1 \rangle & 0 & 0\\
        0 & \ddots & 0\\
        0 & 0 & \langle F, \nabla K^d \rangle
    \end{bmatrix}.
    \end{equation}
\end{proposition}

\autoref{pdelemma} shows that the $d$ components of the generator of a virtual numeraire must satisfy an overdetermined system of $d^2$ partial differential equations. Informally, the equations \eqref{overdetpde} state each transformation $\varphi_t$ must preserve all local, geometric information about every indifference curve, such as the utility-derivative marginal rates of substitution of $\bar{u}$, and any structural relations between them. This may be viewed as a generalization of the classical fact that quasilinear preferences are characterized by the independence of their marginal rates of substitution on the level of numeraire.\footnote{See \cite{blackorby1978extension} and \cite{sono1961effect}. See also \autoref{quasilinearizationthm}.}

In general, no non-zero solution to \eqref{overdetpde} will exist, unless the system satisfies strong compatibility, or `integrability' conditions. 
For any distinct $i, j_0 = 1,\ldots, d$, let $Q^i_{j_0} = K^i_{j_0}$, and for any tuple of indices $(i, j_0, \ldots, j_L)$ with $i \neq j_0$ and $L \ge 0$, recursively define:
\[
    Q^i_{j_0\ldots j_{L+1}} = \frac{\partial Q^i_{j_0\ldots, j_L}}{\partial u^{j_{L+1}}}  + Q^i_{j_0\ldots, j_L}\sum_{l=0}^L K^{j_l}_{j_{L+1}}.
\]
Let $\mathcal{M}_L$ denote the matrix of functions whose rows are given by the vectors:
\begin{equation}\label{compatcond}
     \bigg[\nabla Q^i_{j_0\ldots j_{L'}} + Q^i_{j_0\ldots j_{L'}}\sum_{l=0}^{L'} \nabla K^{j_l}\bigg]^\intercal,
\end{equation}
for each tuple $(i, j_0,\ldots, j_{L'})$ with $i \neq j_0$ and $L' \le L$. Since the rank of any $\mathcal{M}_L$ cannot exceed $d$, there is necessarily some unique, smallest $L^*$ such that, evaluated at $x_0$, the ranks of $\mathcal{M}_{L^*}$ and $\mathcal{M}_{L^*+1}$ coincide. Our final regularity condition requires the ranks of these matrices to be locally constant.

\begin{itemize}
    \item[(CR.3)] \hypertarget{cr3}{} {\bf Constant System Rank}: For each $0 \le L \le L^* + 1$, the ranks of the matrices $\mathcal{M}_L$ are locally constant about $x_0$.
\end{itemize}

We say that the system \eqref{overdetpde} satisfies the {\bf compatibility conditions} if $\textrm{ker}(\mathcal{M}_{L^*})$, evaluated at $x_0$, contains a strictly positive vector. Our next theorem not only characterizes the existence of common, local virtual numeraires, but also shows that the set of such commodities has a well-defined dimension, and is smoothly parameterized by at most $d = \textrm{dim}(\mathcal{X})$ parameters.

\begin{theorem}\label{commonvnthm}
    Assume (\hyperlink{cr1}{CR.1}) - (\hyperlink{cr3}{CR.3}), and suppose $\mathbb{H}$ contains some locally non-zero function. Then the set of local virtual numeraires for $(\succsim^1,\ldots, \succsim^{d+1})$ is non-empty if, and only if the compatibility conditions:
    \[
        \textrm{\emph{ker}}(\mathcal{M}_{L^*}) \cap \mathbb{R}^d_{++} \neq \varnothing
    \]
    hold at $x_0$. When this is true, the set of local virtual numeraires is non-empty, and depends smoothly on $d- \textrm{\emph{rank}}(\mathcal{M_{L^*}})$ parameters.
\end{theorem}

\autoref{commonvnthm} extends straightforwardly to profiles where $S \neq d+1$, as well as to various cases where (\hyperlink{cr1}{CR.1}) does not hold, such as when a profile's rank is less than $d$, or $\nabla u^{d+1}$ is not a positive linear combination of other gradients, see \autoref{ordthmext}.

Under the hypotheses of \autoref{locseptheorem} and \autoref{commonvnthm}, the set of common local numeraires, when non-empty, forms a finite dimensional family.\footnote{\autoref{ordthmext} treats more general profiles, in which independent utility coordinates or directions tangent to every indifference surface can allow for non-parametric freedom.} Thus, interpersonal comparability presents an exacting requirement. For any sufficiently rich profile, of the infinite-dimensional family of possible local utility representations consistent with the given preferences, only a finite-dimensional subset will be interpersonally comparable.

\autoref{commonvnthm} not only allows us to describe the set of common, virtual numeraires for a given profile, but also to obtain necessary and sufficient conditions on \emph{families} of profiles for the existence of comparable representations. We illustrate this below.

\begin{example}\label{haraex}
    Suppose that $\mathcal{X} = \mathbb{R}^3_{++}$, and consider the preference profile represented by the utilities $u^i(x) = x^i$, for $i = 1,2,3$, and $\bar{u}(x) = g(x^1 + x^2)  + x^3$ for some increasing, strictly concave $g$.\footnote{Given our choices of $u^1, u^2$, and $u^3$, we may think of $\mathcal{X}$ as already being in utility coordinates. Thus, in particular, the utility derivatives of $\bar{u}$ coincide with its partial derivatives with respect to each $x^i$.} Note that varying $g$ generally does not correspond to applying any monotone transform to $\bar{u}$, hence we may regard $g$ as indexing a family of preference profiles.

    We claim any tuple of preferences of this form admits a common, local virtual numeraire about an allocation $x_0$ if, and only if, restricted to some neighborhood of $x_0$, $g$ is the Bernoulli utility of a hyperbolic expected utility preference (\citealt{merton1971optimum}), i.e.\ it satisfies:
    \begin{equation}\label{hara}
        -\frac{g'(s)}{g''(s)} = a + bs,
    \end{equation}
    for constants $a,b \in \mathbb{R}$.\footnote{Examples include both log or power functions (the CRRA family) and exponential (CARA) utilities. Since $g$ is assumed to be increasing and strictly concave, $a+bs$ is strictly positive on the range of values taken by $x^1+x^2$ about $x_0$.} For any such $g$, we will show the set of common, local numeraires depends smoothly on three parameters.\footnote{If, instead, we required $g$ to be only weakly concave, on regions where $g$ is linear, \autoref{locseptheorem} would apply and instead yield four-parameter families of local solutions.}

    To prove this, we will make use of \autoref{pdelemma} and \autoref{commonvnthm}. Firstly, we seek a positive function in $\mathbb{H}$. Letting $s = x^1 + x^2$, since $\bar{u}_1 = \bar{u}_2 = g'(s)$, and $\bar{u}_3 = 1$, the only functions in $\mathbb{H}$ that are positive at any given point are $ \bar{h} = -g''(s)/g'(s)^2$. Thus we obtain:
    \[
        K^1 = K^2 = -\ln{\frac{-g''(s)}{g'(s)}}, \quad \textrm{ and } \quad K^3 =  -\ln{\frac{-g''(s)}{g'(s)^2}}.
    \]
    By \autoref{pdelemma}, it suffices to instead determine whether any vector field $F$, defined in a neighborhood of $x_0$, satisfies \eqref{overdetpde} for these $K$.
    
    Since $K^1 = K^2$, by \eqref{overdetpde} it follows any $F$ must satisfy $F^1_1 = F^2_2$. However, since \eqref{overdetpde} also implies each $F^i_i$ must be a function solely of $x^i$, it follows that $F^1_1 = F^2_2 = \alpha$, for some constant $\alpha$. Integrating, we obtain $F^i = \alpha x^i + \beta^i$, for $i =1,2$.

    Plugging these functional forms for $F^1$ and $F^2$ back into \eqref{overdetpde} yields an ordinary differential equation for $g$ that must be satisfied for any such $X$ to exist:
    \[
        \frac{T'(s)}{T(s)} = \frac{\alpha}{\alpha s + \beta},
    \]
    where $T(s) = -g'(s)/g''(s)$. In the context of risk, this quantity would correspond to the Arrow-Pratt measure of risk tolerance for the expected utility preference with Bernoulli utility $g$. Integrating both sides yields:
    \[
        -\frac{g'(s)}{g''(s)} = as + b,
    \]
    as desired. Thus for any common virtual numeraire to exist, $g$ must necessarily obey \eqref{hara}.

    Conversely, to show this is not only necessary but also sufficient for the existence of some common, local numeraire, we apply \autoref{commonvnthm}.  By direct calculation, satisfaction of \eqref{hara} implies each vector of the form \eqref{compatcond} is uniformly zero. Thus $L^* = 0$, and $\mathcal{M}_{L^*}$ is simply the zero matrix, and hence satisfies the integrability conditions. As a consequence, \autoref{commonvnthm} guarantees the set of common, local virtual numeraires is non-empty, and depends smoothly on $3$ parameters, as desired.\hfill $\blacksquare$
    \end{example}

\subsection{\autoref{commonvnthm} for General Profiles}\label{ordthmext}

The preceding constructions extend to profiles of arbitrary finite size and rank by reducing to utility coordinates and combining solutions for smaller sub-profiles. Throughout, we work locally about $x_0$, restricting neighborhoods when necessary.

\subsubsection{Reduction to utility coordinates}

Suppose the profile has (locally) constant rank $r$. After relabeling, choose
representations $u^1,\ldots,u^r$ with independent differentials. Theorem 4.12 of \cite{Lee2002} gives local coordinates:
\[
(q,w)=(u^1,\ldots,u^r,w_1,\ldots,w_{d-r}).
\]
Every remaining utility can then be written as $v^s(q)$, as its
differential lies in the span of the $dq_1,\ldots,dq_r$. By \autoref{liederivformulation}, every common local numeraire has a generator of the form:
\[
\dot{\varphi}=
\sum_{i=1}^{r}f_i(q_i)\frac{\partial}{\partial q_i}
+
\sum_{\alpha=1}^{d-r}
G_\alpha(q,w)\frac{\partial}{\partial w_\alpha},
\]
where each $f_i$ is smooth and positive, and the $G_\alpha$ are arbitrary
smooth functions. For every $s>r$, the remaining invariance and
monotonicity requirements are precisely:
\[
\sum_{i=1}^{r}f_i(q_i)\,\partial_i v^s(q)
=
h_s(v^s(q))
\]
for some smooth positive one-variable function $h_s$.

Conversely, integrating $1/f_i$ and $1/h_s$ gives increasing transformations of the original utilities whose directional derivatives along $\dot{\varphi}$ are equal to unity. The maximal (local) flow of $\dot{\varphi}$ is therefore a common utility unit. The non-parametric freedom captured by the functions $G_\alpha$ describes transformations tangent to each individual's indifference surfaces, and hence leaves fixed each individual's utility. In particular, when $S=r$, any choice of positive $f_i$ is valid.

\subsubsection{Combining lower-dimensional subprofiles}

We now explain how to leverage our results solve these compatibility
equations about sufficiently regular points. Assume each $\partial_i v^s$ either
vanishes identically or is nowhere zero on the neighborhood, and define $J_s=\{i:\partial_i v^s\not= 0\}$. In particular, $v^s$ depends only on the coordinates $q_{J_s}$. If $J_s=\{i\}$, a positive derivative imposes no further restriction on
$f_i$: the directional derivative $f_i(q_i)(v^s)'(q_i)$ is automatically
a positive function of $v^s$.\footnote{However, a negative derivative rules out any common numeraire, as the two preferences require opposite directions of
improvement.}

Now, if $\vert J_s \vert\ge2$, consider the sub-profile consisting of the coordinate preferences $q_i$, $i\in J_s$, together with the preference represented
by $v^s$. On its $\vert J_s\vert $-dimensional coordinate domain, suppose each partial of $v^s$ is positive. Then (\hyperlink{cr1}{CR.1}) obtains, and hence under the the requisite regularity properties, applying \autoref{locseptheorem} in the locally separable case, or \autoref{commonvnthm} otherwise, yields the desired local numeraires; these may then be expressed in the original $q$ coordinates.

For a dependence with mixed signs, simply reversing the coordinate orientations to make each partial of $v^s$ positive does not affect any of its invariance conditions. The sole change then is that monotonicity must be imposed in the original orientations $f_i>0$ and $\sum_{i\in J_s}f_i\,\partial_i v^s>0$.

Finally, common local numeraires for the entire profile can then be constructed by requiring the component functions supplied by different
subproblems to agree pointwise on shared coordinates, restricted to
a sufficiently small common neighborhood. Every common numeraire for the full profile restricts to such subproblem solutions, and conversely, matching these
components and satisfying the positivity conditions yields a
generator which satisfies all of the above compatibility restrictions. Thus this procedure may be used to obtain all common, local numeraires under the sustained regularity hypotheses.

\subsection{Proofs}

\subsubsection{Preliminary Lemmas}

In all that follows, $f^* \omega$ denotes the pullback of the 1-form $\omega$ under the smooth function $f$. We say two 1-forms are proportional if they differ by a positive, smooth, multiple.

\begin{lemma}\label{noninfinitesimalsymmchar}
    Let $\succsim$ be a preference on $\mathcal{X}$ with utility $u$, and  $U,V \subseteq \mathcal{X}$ connected open sets, where the restriction of every indifference surface of $\succsim$ to $U$ is connected. Then for any diffeomorphism $\phi: U \to V$, and any $1$-form $\omega$ proportional to $du$, the following are equivalent:
    \begin{itemize}
        \item[(i)] For all $x,x' \in U$, $x \succsim x'$ if and only if $\phi(x) \succsim \phi(x')$; and
        \item[(ii)] For some smooth $g: U \to \mathbb{R}_{++}$, we have $\phi^* \omega = g \, \omega$.
    \end{itemize}
\end{lemma}
\begin{proof}
    Since $\phi: U \to V$ is a diffeomorphism, (i) is equivalent to $u$ and $u \circ \phi$ representing the same preference on $U$. Thus (i) is equivalent to the existence of some strictly increasing map $f: u(U) \to \mathbb{R}$ such that $u \circ \phi = f \circ u$ on $U$.
    
    We now claim $f$ is smooth. Pick any point $x \in U$. As $u$ is critical-point-free, by the Rank Theorem (\citealt{Lee2002} Theorem 4.12), there exists some neighborhood $x \in W \subseteq U$, open interval $I \subseteq \mathbb{R}$, open ball $B \subseteq \mathbb{R}^{d-1}$, and diffeomorphism $h: I \times B \to W$ such that $u\big(h(t,b)\big) = t$. Thus, $(u \circ \phi \circ h)(t,b)$ is smooth, and since $u$ is constant in $b$, and $u \circ \varphi$ represents the same preference, $(u \circ \phi \circ h)(t,b)$ is constant in $b$. Thus $f$ is simply the restriction of $u \circ \phi \circ h$ to any slice of the form $I \times \{b\}$, and hence smooth. Since smoothness is local, this implies $f$ is smooth on $u(U)$.

    Now, differentiating $u \circ \phi = f \circ u$ at any $x \in W$, we obtain $du_{\phi(x)}\, d\phi_x = f'(u(x)) \, du_x$. Since $\phi$ is a diffeomorphism, $d\phi$ is invertible; since $u$ is critical-point-free, $du \neq 0$, hence $f'\big(u(x)\big) \neq 0$; since $f$ is increasing, $f'\big(u(x)\big) > 0$. Thus (i) is equivalent to $u \circ \phi = f\circ u$ for some strictly increasing, critical-point-free smooth map $f$.

    Now, suppose $u \circ \phi = f\circ u$ for such an $f$, and let $\omega\vert_U = a \, du$ for some smooth $a: U \to \mathbb{R}_{++}$, and likewise $\omega \vert_V = b\, du$ for smooth $b: V \to \mathbb{R}_{++}$. Then by direct calculation:
    \[
        \phi^* \omega = \phi^*(b\, du) = (b\circ \phi) \, d(u \circ \phi) = (b\circ \phi) \, d(f\circ u) = \frac{(b \circ \phi)\, f'(u)}{a} \, \omega,
    \]
    and hence (i) implies (ii). Conversely, suppose $\phi^* \omega = (b \circ \phi) d(u \circ \phi) = ga \, du$, where everywhere $g > 0$. Let $x \sim y$. Since the indifference sets of $\succsim$ are connected, there exists some smooth path $\gamma: [0,1] \to U$ such that $\gamma(0) = x$, $\gamma(1) = y$, and $\gamma$ is contained in the indifference surface through $x$ and $y$. Since for all $t$, $\gamma'(t) \in \textrm{ker}\, du = \textrm{ker}\,  d(u \circ \phi)$, by the Gradient Theorem, $u \circ \varphi$ is constant on the level sets of $u$, thus globally we obtain some function $f: u(U) \to \mathbb{R}$ such that $f \circ u = u \circ \phi$. By an identical argument to before, $f$ is smooth; since $u \circ \phi = f \circ u$, $f'(u) \, du = d(f \circ u) = d(u \circ \phi) = \frac{ga}{b \circ\phi} \, du$, and hence $a,b,g > 0$ implies $f'(u) > 0$. Therefore $f$ is strictly increasing and critical-point-free and thus (ii) implies (i).
\end{proof}

Now, suppose that (\hyperlink{cr1}{CR.1}) holds for the profile $(\succsim^1,\ldots, \succsim^{d+1})$. Let $u^1$,$\ldots$, $u^d,$ and $ \bar{u}$ denote arbitrary choices of representation for each preference. Then the 1-form:
\[
    d\bar{u} = \sum_{i=1}^d \bar{u}_i du^i,
\]
for positive-valued utility derivatives $\bar{u}_1, \ldots, \bar{u}_d$. (Throughout, we do not employ Einstein summation convention.) For each $i = 1,\ldots, d$, define:
\[
    \omega^i = \bar{u}_i du^i,
\]
\and $\bar{u}_i \equiv \partial/\partial u^i \, \bar{u}$. By (\hyperlink{cr1}{CR.1}) these one-forms define a local coframe near $x_0$. 

\begin{lemma}\label{noninfinitesimalsymmchar2}
     Let $U,V \subseteq \mathcal{X}$ be connected open sets such that the restriction of each indifference surface of $\succsim^1,\ldots, \succsim^{d+1}$ to $U$ is connected, and (\hyperlink{cr1}{CR.1}) holds on $U$ and $V$. Then for any diffeomorphism $\phi: U \to V$, the following are equivalent:
    \begin{itemize}
        \item[(i)] For all $x,x' \in U$, we have $x \succsim^i x'$ if and only if $\phi(x) \succsim^i \phi(x')$, for all $i = 1,\ldots, d+1$; and
        \item[(ii)] For some common, smooth $g: U \to \mathbb{R}_{++}$, we have $\phi^* \omega^i = g \, \omega^i$ for all $i = 1,\ldots, d$.
    \end{itemize}
\end{lemma}
\begin{proof}
    By \autoref{noninfinitesimalsymmchar}, there exist smooth functions $g^i : U \to \mathbb{R}_{++}$ such that, on $U$, $\phi^* \omega^i = g^i \omega^i$ and (\hyperlink{n1}{N.1}) for $\succsim^i$ coincide. Similarly, there exists $\bar{g}: U \to \mathbb{R}_{++}$ such that $\phi^* \, d\bar{u} = \bar{g} \, d\bar{u}$ is equivalent to (\hyperlink{n1}{N.1}) for the restriction of $\succsim^{d+1}$ to $U$. Thus, on $U$, (i) implies:
	\[
		\begin{aligned}
			\phi^*(d\bar{u}) & = \phi^* \bigg[ \sum_{i=1}^d \bar{u}_i \, du^i\bigg]\\
			& = \sum_{i=1}^d g^i \omega^i,
		\end{aligned}
	\]
	where, by definition, $d\bar{u} = \sum_i \omega^i$. But since we also have $\phi^* d\bar{u} = \bar{g} \, d\bar{u}$, and $\omega^1, \ldots, \omega^d$ form a coframe on $U$, we must have $g^i = \bar{g}$ for all $i=1 ,\ldots, d$. Hence, on $U$, (i) and (ii) coincide.
\end{proof}

\begin{lemma}\label{liederivformulation}
    Suppose (\hyperlink{cr1}{CR.1}), and let $(\mathcal{D}, \varphi)$ be a virtual commodity defined on some connected open set $U \subseteq \mathcal{X}$ such that the restriction of each indifference surface of $\succsim^1,\ldots, \succsim^{d+1}$ to $U$ are connected. Define $X = \frac{\partial \varphi}{\partial t}\big\vert_{t=0}$. Then $(\mathcal{D}, \varphi)$ satisfies (\hyperlink{n1}{N.1}) and (\hyperlink{n2}{N.2}) for every $\succsim^1,\ldots, \succsim^{d+1}$ if, and only if, for all $i = 1,\ldots, d$, (i) $\mathcal{L}_X \omega^i = \lambda \, \omega^i$, for some smooth real-valued $\lambda$, and (ii) $X^i > 0$.
\end{lemma}
\begin{proof}
    Suppose first $(\mathcal{D}, \varphi)$ is a local virtual numeraire defined on $U$ and let $x \in U$. As $\mathcal{D}^{(x)}$ is an interval, if $t' > t > 0$, then $\mathcal{D}_t \supseteq \mathcal{D}_{t'}$, with the opposite inclusion for negative $t,t'$. Now, by definition, $\mathcal{D}^{(x)}$ contains some non-degenerate closed interval $[-\varepsilon, \varepsilon]$. Thus $\bigcap_{t \in [-\varepsilon,\varepsilon]} \mathcal{D}_t = \mathcal{D}_{-\varepsilon} \cap \mathcal{D}_{\varepsilon}$ is open and contains $x$. By possibly restricting this set further, we may, without loss of generality, suppose $W =  \mathcal{D}_{-\varepsilon} \cap \mathcal{D}_{\varepsilon}$ is a connected neighborhood of $x$ such that the restrictions of each indifference surface of every preference $\succsim^1,\ldots, \succsim^{d+1}$ to $W$ are also connected.
    
    Trivially by (\hyperlink{n2}{N.2}) we must have $X^i > 0$. Now, by \autoref{noninfinitesimalsymmchar2}, for every $t \in [-\varepsilon ,\varepsilon]$, since (\hyperlink{n1}{N.1}) holds, there exist some smooth $g^t: W \to \mathbb{R}_{++}$ such that, on $W$, $\varphi_t^* \, \omega^i = g^t \, \omega^i$ for all $i = 1,\ldots, d$. Since $\varphi$ is smooth in $(t,x)$ so too is $\varphi^*_t\, \omega^i$ and hence so is $g^t \, \omega^i$. Thus, let:
    \[
        \lambda(x) = \frac{\partial g^t(x)}{\partial t}  \bigg\vert_{t=0}.
    \]
    By definition of the Lie derivative of a 1-form (\citealt{Lee2002}, Equation 12.8):
    \[
        \mathcal{L}_X(\omega^i) = \frac{\partial}{\partial t} \varphi^*_t \, \omega^i = \frac{\partial}{\partial t} g^t \omega^i = \lambda \, \omega^i
    \]
    for all $i = 1,\ldots, d$. Since (i) and (ii) are local properties and $x \in U$ was arbitrary, we have shown (\hyperlink{n1}{N.1}) and (\hyperlink{n2}{N.2}) imply (i) and (ii).

    Conversely, suppose $X$ satisfies (i) and (ii) on $U$. Then writing $a = \bar{u}_i \vert_U$, we have:
\[
    \begin{aligned}
        0 & = \mathcal{L}_X (a \, du^i) \wedge (a \, du^i) \\
        & = a X \cdot \nabla a \, du^i \wedge du^i + a \, \mathcal{L}_X du^i \wedge (a \, du^i)\\
        & = (0) + a^2 \, d (X \cdot \nabla u^i) \wedge du^i.
    \end{aligned}
\]
Now, since $a > 0$ by (\hyperlink{cr1}{CR.1}), dividing off yields:
\[
    0 = d (X \cdot \nabla u^i) \wedge du^i.
\]
Thus $d (X \cdot \nabla u^i)$ is collinear with $du^i$. Since it is also exact, and the level sets of $u^i$ connected, by an analogous argument to the proof of \autoref{noninfinitesimalsymmchar}, there exists some smooth, real-valued function $f^i: u^i(U) \!\to \mathbb{R}$ such that $X \cdot \nabla u^i = f^i \circ u^i$.  By hypothesis, $X \cdot \nabla u^i = du^i(X) = X^i > 0$, hence $f^i > 0$. Thus the function $F^i: u^i(U) \to \mathbb{R}$ given by:
\[
    F^i(q) = \int_{\bar{q}}^q \frac{1}{f^i(\xi)} \, d\xi,
\]
for some $\bar{q} \in u^i(U)$, is well-defined and strictly increasing.  Then:
\[
    \begin{aligned}
        X \cdot \nabla F^i(u^i) & = \sum_j X^j {F^i}'(u^i) \frac{\partial}{\partial x^j} u^i\\
        & = \frac{1}{f^i(u^i)} \sum_j X^j  \frac{\partial}{\partial x^j} u^i\\
        & = \frac{1}{f^i(u^i)} X \cdot \nabla u^i\\
        & = 1.
    \end{aligned}
\]
Thus $F^i \circ u^i$ is both a strictly increasing transformation of $u^i$ (and hence represents $\succsim^i$ on $U$), and has constant $X$-directional derivative equal to unity, as desired. Let $(\mathcal{D},\varphi)$ then denote the maximal flow on $U$ associated with $X$; it follows that $F^i \circ u^i \vert_U$ has $(\mathcal{D},\varphi)$ as a utility unit, and hence $(\mathcal{D},\varphi)$ is a virtual numeraire for $\succsim^i$, for all $i= 1,\ldots, d$. Finally, since:
\[
    \mathcal{L}_X(d\bar{u}) = \sum_i \lambda\,  \omega^i = \lambda \, d\bar{u},
\]
an analogous argument shows $(\mathcal{D},\varphi)$ is also a virtual numeraire for $\succsim^{d+1}$, completing the proof. 
\end{proof}

Thus locally, the satisfaction of (\hyperlink{n1}{N.1}) is equivalent to a diffeomorphism acting on the $d$ 1-forms $\omega^1, \ldots, \omega^d$ by a common, multiplicative scaling. We now expand our space to incorporate this added degree of freedom explicitly. Given $x_0 \in \mathcal{X}$ and $U,V \subseteq \mathcal{X}$ connected open sets with $x_0 \in U$, we define the {\bf lifted} 1-forms $\theta^i$ on $\mathbb{R}_{++} \times U$ and $\mathbb{R}_{++} \times V$ via:
\[
    \theta^i \vert_{(h,u)} = h \, \omega^i \vert_u.
\]
The following lemma shows that a diffeomorphism $\phi$ fixes every $\omega^i$, modulo a common scaling, if and only if there exists a `lifted' diffeomorphism $\Phi$ which fixes each $\theta^i$.

\begin{lemma}\label{lifteddiff}
     Let $U,V \subseteq \mathcal{X}$ be connected open sets. Then there exists a diffeomorphism $\phi: U \to V$ such that, for all $i=1,\ldots, d$, $\phi^* \omega^i = g\, \omega^i$ for some smooth $g: U \to \mathbb{R}_{++}$  if, and only if, there exists a diffeomorphism $\Phi:\mathbb{R}_{++} \times U \to \mathbb{R}_{++} \times V$ such that, for all $i=1,\ldots, d$, we have $\Phi^* \theta^i = \theta^i$.
\end{lemma}
\begin{proof}
Suppose first there exists a diffeomorphism $\Phi : \mathbb{R}_{++} \times U \to  \mathbb{R}_{++} \times V$  satisfying $\Phi^* \theta^i = \theta^i$ for  $i = 1, \ldots, d$. Write $\Phi  = \big(\tilde{h}(h,u), \tilde{\phi}(h, u) \big)$. By direct computation:
\[
	\begin{aligned}
		\Phi^* \theta^i & = \tilde{h}(h,u) \tilde{\phi}^* \omega^i,
	\end{aligned}
\]
hence:
\begin{equation}\label{liftedpullbackderiv}
	\tilde{\phi}^* \omega^i =  \frac{h}{\tilde{h} }  \, \omega^i.
\end{equation}
Now on the left-hand side of \eqref{liftedpullbackderiv}, 
\[
	\tilde{\phi}^* \omega^i =  \bar{u}_i(\tilde{\phi}(h,u)) \bigg[\sum_{j=1}^d \frac{\partial \tilde{\phi}^i}{\partial u^j} \, du^j  + \frac{\partial \tilde{\phi}^{i}}{\partial h} \, dh \bigg].
\]
But on the right-hand side of \eqref{liftedpullbackderiv}, since each $\omega^i$ has a 0 coefficient on $dh$, the $h$-partial derivative of $\tilde{\phi}$ must uniformly vanish:
\[
	\frac{\partial \tilde{\phi}}{\partial h} = 0.
\]
Since $U$ is connected, this implies $\tilde{\phi}$ is a function only of the utility variables $u = (u^1,\ldots, u^d)$.  Thus the entire left-hand side of \eqref{liftedpullbackderiv} is a function solely of the $u$ variables, and thus so too is the right.  This in turn implies $h/\tilde{h}(h,u)$ is a function $g(u)$, i.e.\ that $h / g(u) = \tilde{h}(h,u)$, and hence the map $\Phi(h,u) = (h/g(u), \tilde{\phi}(u)\big)$. In particular, since $h$ and $\tilde{h}(h,u)$ are both strictly positive, so too is $g$, and $\tilde{\phi}^* \omega^i = g \,\omega^i$. Finally, in light of the above, it follows $\tilde{\phi}$ must map $U \to V$ diffeomorphically, as $\Phi$ is a diffeomorphism itself.

Conversely, suppose $\phi: U \to V$ is a diffeomorphism such that $\phi^* \omega^i = g \, \omega^i$ for all $i = 1, \ldots d$.  Define:
\[
	\Phi(h,u) = \bigg(\frac{h}{g(u)}, \phi(u)\bigg).
\]
Then:
\[
	\begin{aligned}
		\Phi^* \theta^i & = \frac{h}{g(u)} \phi^\ast \omega^i \\
		& = h \frac{g}{g} \omega^i = h \omega^i = \theta^i
	\end{aligned}
\]
as desired.
\end{proof}

\begin{lemma}\label{invariantcoframelemma}
    Suppose (\hyperlink{cr1}{CR.1}) and (\hyperlink{cr2}{CR.2}) hold, and $d \ge 3$. Then for any connected, open $U,V \subseteq \mathcal{X}$, and any $h_0 \in \mathbb{H}$ such that $h_0 \vert_U > 0$, a diffeomorphism $\phi: U \to V$ satisfies $ \phi^* \omega^i = g \, \omega^i$ for all $i =1,\ldots, d$ and some smooth $g: U \to \mathbb{R}_{++}$ if, and only if, $\phi^* h_0 \, \omega^i = h_0 \, \omega^i$ for all $i = 1,\ldots, d$.
\end{lemma}
\begin{proof}
We begin by considering the lifted coframe elements $\theta^i = h \, \omega^i$. By direct computation:
\begin{equation}\label{se}
	\begin{aligned}
		d\theta^i & = d h \omega^i\\
		& = dh \wedge \omega^i + h d\omega^i\\
		& = \frac{dh}{h} \wedge \theta^i + h \sum_{j = 1}^d\frac{ \bar{u}_{ij} }{\bar{u}_i \bar{u}_j} \omega^j \wedge \omega^i\\
		& = \frac{dh}{h} \wedge \theta^i + \frac{1}{h} \sum_{j = 1}^d\frac{ \bar{u}_{ij} }{\bar{u}_i \bar{u}_j} \theta^j \wedge \theta^i.
	\end{aligned}
\end{equation}
Consider the restriction of these forms to the graph of some function $\hat{h}: U \to \mathbb{R}_{++}$. Once $\hat{h}$ is specified, since $\theta^1 ,\ldots, \theta^d$ form a coframing of $U$, we have:
\[
	\frac{d\hat{h}}{\hat{h}} = z_1 \theta^1 + \cdots + z_d \theta^d,
\]
for some unknown smooth, real-valued functions $z_1, \ldots, z_d$ defined on $U$, thus \eqref{se} becomes:
\[
	d\theta^i = \sum_{j\neq i} \bigg[ z_j + \frac{1}{\hat{h}} \frac{ \bar{u}_{ij} }{\bar{u}_i \bar{u}_j}\bigg] \theta^j \wedge \theta^i.
\]
Now, suppose $\phi: U \to V$ is a diffeomorphism satisfying $\phi^* \omega^i = g \, \omega^i$ for all $i =1,\ldots, d$, for some positive-valued smooth function $g$. By \autoref{lifteddiff}, the lifted diffeomorphism $\Phi: \mathbb{R}_{++} \times U \to \mathbb{R}_{++} \times V$, defined by $\Phi(h,u) = \big(h/g(u), \phi(u)\big)$, satisfies $\Phi^* \theta^i = \theta^i$ for all $i=1,\ldots, d$.  Thus on the graph of $\hat{h}$:
\[
	\begin{aligned}
		\Phi^*\big(d\theta^i\big) & = d\, \Phi^* \theta^i \\
		& = d \theta^i\\
		& =  \sum_{j\neq i} \bigg[ z_j + \frac{1}{\hat{h}} \frac{ \bar{u}_{ij} }{\bar{u}_i \bar{u}_j}\bigg] \theta^j \wedge \theta^i.
	\end{aligned}
\]
But once again, since (i) $\Phi^* \theta^i = \theta^i$ for all $i$, and (ii) $\Phi(\hat{h}(u),u) = \big(\hat{h}(u)/g(u), \phi(u)\big)$, we obtain:
\begin{equation}\label{esstor}
	\tilde{z}_j\big \vert_{\phi(u)} + \frac{g(u)}{\hat{h}(u)} \cdot \frac{\bar{u}_{ij}}{\bar{u}_i\bar{u}_j} \bigg\vert_{\phi(u)} = z_j\big \vert_{u} + \frac{1}{\hat{h}(u)} \frac{\bar{u}_{ij}}{\bar{u}_i\bar{u}_j} \bigg\vert_{u}
\end{equation}
for all $i \neq j$, where $\tilde{z}$ denotes the analogous $z$ functions obtained for the lifted forms on $\mathbb{R}_{++} \times V$. As such, let $i,j,k$ be distinct. Then for every symmetry $\phi$, we obtain:
\begin{equation}\label{esstor2}
	\bigg[ \frac{1}{h} \frac{\bar{u}_{ij}}{\bar{u}_i\bar{u}_j}  - \frac{1}{h} \frac{\bar{u}_{kj}}{\bar{u}_k\bar{u}_j}\bigg]_{\Phi(\hat{h}(u),u)}   = \bigg[ \frac{1}{h} \frac{\bar{u}_{ij}}{\bar{u}_i\bar{u}_j}  - \frac{1}{h} \frac{\bar{u}_{kj}}{\bar{u}_k\bar{u}_j}\bigg]_{(\hat{h}(u),u)}
\end{equation}
by subtracting \eqref{esstor} for indices $i$ and $j$, from the analogous equation for indices $k$ and $j$. By hypothesis, for some choice of distinct $i,j,k$, both sides of \eqref{esstor2} are strictly positive for all $u \in U$.\footnote{This is possible as $d \ge 3$.} Define, now, for these indices, the functions $T: \mathbb{R}_{++} \times U \to \mathbb{R}_{++}$ and $h_0 : U \to \mathbb{R}_{++}$ via:
\[
    T(h, u) = \frac{1}{h} \bigg[ \frac{\bar{u}_{ij}}{\bar{u}_i\bar{u}_j} -  \frac{\bar{u}_{kj}}{\bar{u}_k\bar{u}_j}\bigg],
\]
and:
\[
    h_0(u) = \bigg[\frac{\bar{u}_{ij}}{\bar{u}_i\bar{u}_j}  -  \frac{\bar{u}_{kj}}{\bar{u}_k\bar{u}_j}\bigg].
\]
Clearly $h_0 \in \mathbb{H}$ and, by construction:
\begin{equation}\label{hnaughtident}
    1 = T\big(h_0(u), u\big).
\end{equation}
This implies that, for any $\phi$ and smooth, positive-valued function $g$ such that $\phi^* \omega^i = g\, \omega^i$ for all $i=1,\ldots, d$, we have:
\begin{equation}\label{esstorinv}
    \begin{aligned}
        1 & = \frac{(T \circ \Phi)(h_0(u),u)}{T(h_0(u), u)} \\
        & = \frac{(T \circ \Phi)(h_0(u),u)}{1} \\
        & =  \frac{(h_0 \circ \phi)(u)g(u)}{h_0(u)},
    \end{aligned}
\end{equation}
where the first equality follows from dividing the left by the right-hand side of \eqref{esstor2}, the second from \eqref{hnaughtident}, and the final by definition of $\Phi$. Consider now the modified coframe $h_0 \, \omega^i$, for $i = 1,\ldots, d$. Under $\phi$ we have:
\[
    \phi^* h_0 \,\omega^i = (h_0\circ \phi) \, g\,  \omega^i = \bigg[\frac{(h_0\circ \phi) \, g}{h_0}\bigg] \, h_0 \, \omega^i.
\]
By \eqref{esstorinv}, the square bracketed term is the constant function $1$, hence $\phi^* h_0 \, \omega^i = h_0 \, \omega^i$ for all $i = 1, \ldots, d$. Since the converse is trivial, the claim follows.
\end{proof}

\noindent When $d=2$, we define $\mathbb{H}$ to be the singleton set consisting of the function:
\[
        h_0(u) = \sqrt{\bigg\vert\frac{1}{\bar{u}_1 \bar{u}_2} \partial_{12} \ln{\bigg( \frac{\bar{u}_1}{\bar{u}_2}\bigg)} \bigg\vert}.
\]
In such cases, (\hyperlink{cr2}{CR.2}) requires the sign of this function to be locally constant.

\begin{lemma}\label{invariantcoframelemmad=2}
    Suppose (\hyperlink{cr1}{CR.1}) and let $d=2$. Then for any connected, open $U,V \subseteq \mathcal{X}$ such that for the unique $h_0 \in \mathbb{H}$, we have $h_0 \vert_U > 0$, a diffeomorphism $\phi: U \to V$ satisfies $\phi^* \omega^i = g \, \omega^i$ for all $i = 1,2$ and some smooth $g: U \to \mathbb{R}_{++}$ if, and only if, $\phi^* h_0 \,\omega^i = h_0 \,\omega^i$ for all $i= 1,2$.
\end{lemma}
\begin{proof}
    We once again begin by considering the lifted 1-forms $\theta^1$ and $\theta^2$. By direct computation:
    \[
        \begin{aligned}
            d \begin{bmatrix}
                \theta^1 \\ \theta^2
            \end{bmatrix} & = d \begin{bmatrix}
                h \, \omega^1 \\ h \, \omega^2
            \end{bmatrix}\\
                & = \begin{bmatrix}
                    \frac{dh}{h} & 0 \\ 0 & \frac{dh}{h}
                \end{bmatrix} \wedge \begin{bmatrix}
                    \theta^1 \\ \theta^2 
                \end{bmatrix} + \begin{bmatrix}
                   - \frac{1}{h} \frac{\bar{u}_{12}}{\bar{u}_1 \bar{u}_2}\\
                   \frac{1}{h} \frac{\bar{u}_{12}}{\bar{u}_1 \bar{u}_2}
                \end{bmatrix}  \theta^{1} \wedge \theta^{2}.
        \end{aligned}
    \]
    Now, let $\kappa = \frac{dh}{h} + \frac{1}{h} \frac{\bar{u}_{12}}{\bar{u}_1 \bar{u}_2} \theta^1  + \frac{1}{h} \frac{\bar{u}_{12}}{\bar{u}_1 \bar{u}_2} \theta^2  $. Then $d \theta^i = \kappa  \wedge \theta^i$. Differentiating this, we obtain that for $i= 1,2$:
    \[
        \begin{aligned}
            0 & = d \kappa  \wedge \theta^i - \kappa  \wedge d\theta^i\\
            & = d \kappa  \wedge \theta^i - \kappa  \wedge \kappa  \wedge \theta^i\\
            & = d\kappa  \wedge \theta^i.
        \end{aligned}
    \]
    As $U \subseteq \mathcal{X}$ which is two-dimensional, this implies $d\kappa  = c \, \theta^1 \wedge \theta^2$ for some smooth function $c$.
    
    Now, suppose $\phi$ satisfies $\phi^* \omega^i = g \, \omega^i$. By \autoref{lifteddiff}, the map $\Phi(h,u) = \big(h/g(u), \phi(u)\big)$ satisfies $\Phi^* \theta^i = \theta^i$. Then:
    \[
        \begin{aligned}
            \Phi^* d\, \theta^i & = d \Phi^* \theta^i \\
            & = d \theta^i \\
            & = \kappa \wedge \theta^i.
        \end{aligned}
    \]
    As pullbacks commute with wedge products, we obtain that in addition, $\Phi^* \kappa = \kappa$. From this fact, and our observation that $d \kappa = c \, \theta^1 \wedge \theta^2$, we conclude that  $c \circ \Phi(h,u) = c(h,u)$.
    
    Now, by direct computation, writing $T$ for $\frac{\bar{u}_{12}}{\bar{u}_1 \bar{u}_2}$, we have:
    \begin{equation}\label{2dcurvature}
        \begin{aligned}
            d \kappa & = d\bigg[ \frac{dh}{h} + T \omega^1 + T  \omega^2 \bigg] \\
            & = d T \wedge \omega^1 + T \,d\omega^1 + dT \wedge \omega^2 + T\, d\omega^2\\
            & = \frac{T_2}{\bar{u}_2} \omega^2 \wedge \omega^1 + \frac{T_1}{\bar{u}_1} \omega^1 \wedge \omega^2 \\
            & = \bigg(\frac{T_1}{\bar{u}_1}  - \frac{T_2}{\bar{u}_2}\bigg) \, \omega^1 \wedge \omega^2.
        \end{aligned}
    \end{equation}
   Now, note that $T = \frac{1}{\bar{u}_1} \partial_1 \ln{\bar{u}_2} = \frac{1}{\bar{u}_2} \partial_2 \ln{\bar{u}_1}$, hence:
   \begin{equation} \label{2dcurv1}
        \frac{T_1}{\bar{u}_1} = -T^2 + \frac{1}{\bar{u}_1 \bar{u}_2} \partial_{12} \ln{\bar{u}_1}
   \end{equation}
   and
   \begin{equation} \label{2dcurve2}
        \frac{T_2}{\bar{u}_2} = -T^2 + \frac{1}{\bar{u}_1\bar{u}_2} \partial_{12} \ln{\bar{u}_2}.
   \end{equation}
   In particular, from \eqref{2dcurvature}, \eqref{2dcurv1}, and \eqref{2dcurve2} we have:
    \[
        \begin{aligned}
            c(h,u) & = \frac{1}{h^2}\bigg[\frac{T_1}{\bar{u}_1} - \frac{T_2}{\bar{u}_2}\bigg]\\
            & = \frac{1}{h^2} \frac{1}{\bar{u}_1 \bar{u}_2} \partial_{12} \ln{\bigg( \frac{\bar{u}_1}{\bar{u}_2}\bigg)}.
        \end{aligned}
    \]
    Since $h_0 \vert_U > 0$, we may suppose $c > 0$ for all $h$ and all $u \in U$ and we obtain $c(h_0(u), u) = 1$ for all $u\in U$.\footnote{That is, we may consider both the apparent positive and negative branches together, by simply interchanging $\theta^1$ and $\theta^2$, i.e. $d\kappa = -c \theta^2 \wedge \theta^1$. See \cite{gardner1989method} p.\ 65 for more discussion.}  Now, as $c \circ \Phi(h,u) = c(h,u)$:
    \[
        \begin{aligned}
            1 & = c \circ \Phi\big(h_0(u), u\big)\\
            & = \bigg[ \frac{(h_0 \circ \phi)(u) g(u)}{h_0(u)}\bigg]^2
        \end{aligned}  
    \]
    and hence since $h_0$ and $g$ are positive,
    \[
        \frac{(h_0 \circ \phi)(u) g(u)}{h_0(u)} = 1.
    \]
    But then:
    \[
        \phi^* h_0 \, \omega^i = (h_0\circ \phi)\, g \, \omega^i = \frac{(h_0 \circ \phi) \, g}{h_0} \, h_0 \, \omega^i = h_0 \, \omega^i,
    \]
    as desired.
\end{proof}

\subsection{Proof of \autoref{pdelemma}}

\begin{proof}
    Suppose first $(\mathcal{D}, \varphi)$ is a local virtual numeraire defined on some convex open neighborhood $U$ of $x_0$. As every $\mathcal{D}^{(x)}$ is an interval, if $t' > t > 0$, then $\mathcal{D}_t \supseteq \mathcal{D}_{t'}$, with the opposite inclusion for negative $t,t'$. Now, by definition, $\mathcal{D}^{(x_0)}$ contains some non-degenerate closed interval $[-\varepsilon, \varepsilon]$. Thus $\bigcap_{t \in [-\varepsilon,\varepsilon]} \mathcal{D}_t = \mathcal{D}_{-\varepsilon} \cap \mathcal{D}_{\varepsilon}$ is open and contains $x_0$. By possibly restricting this set further, we may, without loss of generality, suppose $W =  \mathcal{D}_{-\varepsilon} \cap \mathcal{D}_{\varepsilon}$ is connected, and that the restriction of each indifference surface of every preference $\succsim^1,\ldots, \succsim^{d+1}$ to it are connected. Then by \autoref{noninfinitesimalsymmchar2}, for every $t \in [-\varepsilon ,\varepsilon]$, since (\hyperlink{n1}{N.1}) holds, there exist some smooth $g^t: W \to \mathbb{R}_{++}$ such that, on $W$, $\varphi_t^* \omega^i = g^t \omega^i$ for all $i = 1,\ldots, d$, where $\varphi_t^* \omega^i$ denotes the pull-back of the 1-form $\omega^i$ under $\varphi_t$.

    Now, by hypothesis, $\bar{h} \in \mathbb{H}$ satisfies $\bar{h} \vert_W > 0$. Thus by \autoref{invariantcoframelemma} and \autoref{invariantcoframelemmad=2}, regardless of whether $d>2$ or $d=2$, we have $\varphi_t^*(\bar{h}\, \omega^i) = \bar{h} \, \omega^i$ for every $t \in [-\varepsilon, \varepsilon]$. Then, letting $X$ denote the generator of $\varphi$, $\dot{\varphi}$, by definition of the Lie derivative of a 1-form (\citealt{Lee2002}, Equation 12.8):
    \[
        \mathcal{L}_X(\bar{h} \, \omega^i) = \lim_{t \to 0} \frac{ \varphi_t^*(\bar{h} \, \omega^i) - \bar{h} \, \omega^i}{t}
    \]
    is well-defined using only $t \in [-\varepsilon, \varepsilon]$, and equal to zero. But in utility coordinates:
    \[
    \begin{aligned}
        \mathcal{L}_X(\bar{h} \,\omega^i) & = X \cdot \nabla (\bar{h} \bar{u}_i) \, du^i + \bar{h}\, \bar{u}_i \mathcal{L}_X(du^i) \\
        & = X \cdot \nabla (\bar{h} \bar{u}_i) \, du^i + \bar{h} \, \bar{u}_i \sum_{j=1}^d X^i_j \, du^j.
    \end{aligned}
    \]
    Since $du^1,\ldots, du^d$ form a coframe, $\mathcal{L}_X(\bar{h} \, \omega^i) = 0$ if and only if (i) $X^i_j = 0$ for all $i \neq j$, and (ii) $X^i_i = -X \cdot \nabla \ln (\bar{h} \bar{u}_i)$. Writing these equations in matrix form is precisely \eqref{overdetpde}. Moreover, $X^i = X \cdot \nabla u^i = du^i(X)$, for all $i=1, \ldots, d$, hence satisfaction of (\hyperlink{n2}{N.2}) implies $X^i > 0$.
    
    Conversely, suppose \eqref{overdetpde} admits a solution $X$ on some neighborhood $U$ of $x_0$, and that $X^i > 0$ for all $ i= 1,\ldots, d$, and hence (\hyperlink{n2}{N.2}) obtains. By restricting, we may assume the restriction each indifference surface of every preference $\succsim^1,\ldots, \succsim^{d+1}$  to $U$ are connected. We have already shown that satisfaction of \eqref{overdetpde} by $X$ is equivalent to $\mathcal{L}_X(\bar{h} \, \omega^i) = 0$. But $ 0 = \mathcal{L}_X(\bar{h}\, \omega^i) = X \cdot \nabla \bar{h} \,\, \omega^i + \bar{h} \, \mathcal{L}_X(\omega^i)$, hence, as $\bar{h} > 0$, $\mathcal{L}_X(\omega^i) = -X \cdot \nabla \ln(\bar{h}) \; \omega^i$, and thus $\mathcal{L}_X(\omega^i)$ can be written as $\lambda \, \omega^i$, for all $i = 1,\ldots, d$. As a consequence, by \autoref{liederivformulation}, the maximal flow associated with $X$ on $U$ is a local virtual numeraire for $\succsim^1, \ldots, \succsim^{d+1}$.
\end{proof}

\subsection{Proof of \autoref{commonvnthm}}

\subsubsection{Preliminaries}

Let us write the first-order system \eqref{overdetpde} as:
\begin{equation}\label{overdetpde2}
    X^i_j = \psi^i_j(X, u) = \begin{cases}
        X \cdot \nabla K^i(u) & \textrm{ if }  i=j\\
        0 & \textrm{ otherwise,}
    \end{cases}   
\end{equation}
where, as before, $X^i_j$ denotes $\partial X^i /\partial u^j$. Define the differential operators:
\[
    \mathcal{D}_i = \frac{\partial}{\partial u^i} + \psi^i_i \frac{\partial }{\partial X^i},
\]
for all $i = 1,\ldots, d$ on $C^\infty(\mathbb{R}^d \times U)$. If $X$ is a solution to \eqref{overdetpde2}, then for any smooth function $f(X(u), u)$, applying $\mathcal{D}_i$ to $f$ corresponds to totally differentiating $f$ with respect to $u^i$, then substituting $\psi^i_j$ for $X^i_j$ via  \eqref{overdetpde2}.

The key integrability result for establishing the non-emptiness and dimensionality of the set of local solutions to \eqref{overdetpde2} is due to \cite{veblen1926projective}. We give a formal statement of this result, tailored to our system, below.\footnote{The original proof in \cite{veblen1926projective} or \cite{eisenhart1927non} is very terse, and only informally states several needed regularity conditions. We provide a more formal statement here.} Toward this end, for any $i \neq j$ in $1, \ldots, d$, let:
\[
    I^i_j(X,u) = X \cdot \big[\nabla K^i_j(u) + K^i_j(u) \nabla K^j(u) ].
\]
For any family $\mathcal{F} \subset C^\infty(\mathbb{R}^d \times U)$, let $\langle \mathcal{F} \rangle_U$ denote the set of all finite, linear combinations of functions in $\mathcal{F}$, whose coefficients are smooth functions $U \to \mathbb{R}$.  Then, define:
\[
    \mathcal{I}^0 = \big \langle I^i_j(X,u) : i \neq j\big \rangle_U,
\]
and inductively, let: $\mathcal{I}^{L+1} = \big\langle\mathcal{I}^L \cup \{\mathcal{D}_j I : I \in \mathcal{I}^L, j = 1,\ldots, d\} \rangle_U$ for all $L \ge 0$, and let $\mathcal{Z}^L = \big\{(X,u) \in \mathbb{R}^d \times U : I(X,u) = 0 \textrm{ for all } I \in \mathcal{I}^L\big\}$. We note that the sets $\mathcal{Z}^L$ could be equivalently obtained without taking the $C^\infty(U)$-spans, $\langle \, \cdot \, \rangle_U$, as formation of combinations of this type leaves the common zero set unchanged.

\begin{theorem*}[\citealt{veblen1926projective}]
    Let $(X_0, u_0) \in \mathbb{R}^d \times U$ and suppose that, in some neighborhood of $(X_0, u_0)$, there exist $L^* \ge 0$ and some tuple of functions $\hat{I} = (\hat{I}^1, \ldots, \hat{I}^p) \subseteq \mathcal{I}^{L^*}$ such that:
    \begin{itemize}
        \item[(i)] The matrix:
        \[
            \begin{bmatrix}
                \frac{\partial \hat{I}^i}{\partial X^j}
            \end{bmatrix}_{1\le i \le p,\, 1 \le j \le d}
        \]
        has constant rank $p$ in a neighborhood of $(X_0, u_0)$.
        \item[(ii)] $\mathcal{Z}^{L^*} = \mathcal{Z}^{L^*+1}$, and locally about $(X_0, u_0)$, $\mathcal{Z}^{L^*}$ is precisely the common zero set of $\hat{I}^1,\ldots, \hat{I}^p$.
        \item[(iii)] $(X_0, u_0) \in \mathcal{Z}^{L^*}$.
    \end{itemize}
    Then there exists a neighborhood $\mathcal{N}_0$ of $u_0$ and a unique solution $X$ to \eqref{overdetpde2} on $\mathcal{N}_0$, with $X(u_0) = X_0$. Moreover, the family of local solutions is indexed by the $\tilde{X}_0$ such that $(\tilde{X}_0, u_0) \in \mathcal{Z}^{L^*}$ and hence depends smoothly on $d-p$ constants.
\end{theorem*}

Finally, we supply a book-keeping lemma which relates the matrices $\mathcal{M}$ to the functions in each $\mathcal{I}^L$.

\begin{lemma}\label{integrabilitycondlemma}
    For any tuple $\alpha = (i, j_0, \ldots, j_L)$, $i \neq j_0$, let $R_\alpha(X,u) = X \cdot r_\alpha(u)$ and $r_\alpha(u) = \nabla Q^i_{j_0\ldots j_L} + Q^i_{j_0\ldots j_L} \sum_{l=0}^L \nabla K^{j_l}$. Then for any $0 \le \bar{L}$,
    \begin{equation}\label{spanequiv}
        \mathcal{I}^{\bar{L}} = \big\langle R_\alpha : \alpha = (i, j_0, \ldots, j_L), i \neq j_0, L \le \bar{L} \big\rangle_U,
    \end{equation}
    and hence:
    \[
        \mathcal{Z}^{\bar{L}} = \big\{ (X,u) \in \mathbb{R}^d \times U : \mathcal{M}_{\bar{L}}(u) X = 0 \big\}.
    \]
\end{lemma}
\begin{proof}
    Fix $\alpha = (i,j_0, \ldots, j_L)$, $i \neq j_0$. Then, for any $k = 1,\ldots, d$, by direct calculation:
    \[
        \begin{aligned}
            & \mathcal{D}_k \bigg( X \cdot \bigg[\nabla Q^i_{j_0\ldots j_L} + Q^i_{j_0\ldots j_L} \sum_{l=0}^L \nabla K^{j_l} \bigg]\bigg) =   \sum_{a=1}^d X^a \bigg[K^k_a \frac{\partial}{\partial u^k}Q^i_{j_0\ldots j_L} + K^k_a  Q^i_{j_0\ldots j_L}\sum_{l=0}^L  K^{j_l}_k  + \cdots \\
            & \quad \quad \cdots + \frac{\partial^2}{\partial u^a \partial u^k} Q^i_{j_0\ldots j_L} + \bigg(\frac{\partial}{\partial u^k}Q^i_{j_0\ldots j_L}\bigg)\sum_{l=0}^L K^{j_l}_a  +  Q^i_{j_0\ldots j_L} \sum_{l=0}^L K^{j_l}_{ak} \bigg]\\
            & =  \sum_{a=1}^d X^a \bigg[K^k_a Q^i_{j_0\ldots j_L k} + \frac{\partial^2}{\partial u^a \partial u^k} Q^i_{j_0\ldots j_L} +  \cdots \\
            & \quad \quad \cdots + \bigg(Q^i_{j_0 \ldots j_L k} - Q^i_{j_0 \ldots j_L} \sum_{l=0}^L K^{j_l}_k\bigg) \sum_{l=0}^L K^{j_l}_a  +  Q^i_{j_0\ldots j_L} \sum_{l=0}^L K^{j_l}_{ka} \bigg]\\
        \end{aligned}
    \]
    Expanding the cross-derivative of $ Q^i_{j_0\ldots j_L}$ using the definition of $Q^i_{j_0\ldots j_L k}$, and canceling terms, yields:
    \[
        \begin{aligned}
            \cdots &  =  \sum_{a=1}^d X^a \bigg[Q^i_{j_0\ldots j_L k}\bigg(K^{j_0}_a + \cdots + K^{j_L}_a + K^k_a \bigg) + \frac{\partial}{\partial u^a}Q^i_{j_0 \ldots j_L k} +  \cdots \\
            & \quad \quad \cdots - \bigg(\frac{\partial}{\partial u^a} Q^i_{j_0 \ldots j_L}+Q^i_{j_0 \ldots j_L} \bigg(\sum_{l=0}^L K^{j_l}_a\bigg)\bigg)\bigg( \sum_{l=0}^L K^{j_l}_k\bigg)\bigg]\\
            & = X \cdot \bigg[ \nabla Q^i_{j_0\ldots j_L k} +  Q^i_{j_0\ldots j_L k} \big(\nabla K^{j_0} + \cdots + \nabla K^{j_L} + \nabla K^k \big) -  \cdots \\
            & \quad \quad \cdots - \bigg( \sum_{l=0}^L K^{j_l}_k\bigg) \bigg(\nabla Q^i_{j_0 \ldots j_L} + Q^i_{j_0 \ldots j_L} \sum_{l=0}^L \nabla K^{j_l}\bigg) \bigg].
        \end{aligned}  
    \]
    The dot product of $X$ with first summand in the square brackets is simply $R_{\alpha k}$, where $\alpha k = (i, j_0, \ldots, j_L, k)$. Similarly, the dot product of $X$ with the second summand is $\big( \sum_{l=0}^L K^{j_l}_k\big)$ times $R_\alpha$. Thus $R_{\alpha k}(X,u) = \mathcal{D}_k R_{\alpha} + \bigg(\sum_{l=0}^L K^{j_l}_k\bigg) R_{\alpha}$. For $\bar{L} = 0$, the $R_{i j_0}$ are precisely the $I^i_{j_0}$ hence \eqref{spanequiv} holds trivially. Suppose then that \eqref{spanequiv} holds for some $\bar{L} \ge 0$. Then the above calculation (and Leibniz's rule) shows that $\mathcal{D}_k$ of any element of the right-hand side of \eqref{spanequiv} belongs to the span of the functions $R_{\alpha'}$ where $\alpha'$ is of length $\le \bar{L}+1$. Conversely, the above calculation shows that each new generator $R_{\alpha k}$ belongs to the span of $\mathcal{I}^L$ and its derivatives under the operators $\mathcal{D}_k$. Thus \eqref{spanequiv} holds for indices of length $\bar{L} + 1$, and the claim follows by induction. Finally, note the rows of $\mathcal{M}_{\bar{L}}(u)$ are precisely the transposes of the coefficients $r_\alpha$, hence the final claim is immediate. 
\end{proof}

\subsubsection{Proof of \autoref{commonvnthm}}

\begin{proof}
    $(\Longleftarrow)$: Suppose that $\textrm{ker} \mathcal{M}_{L^*}(u_0) \cap \mathbb{R}^d_{++} \neq \varnothing$, and suppose $X_0$ is any arbitrary element of this intersection. Let $p = \textrm{rank} \, \mathcal{M}_{L^*}(u_0)$. Since $X_0 \neq 0$ belongs to this kernel, $p < d$. Now, by definition of $L^*$, $ \textrm{rank}\, \mathcal{M}_{L^*}(u_0) = \textrm{rank}\, \mathcal{M}_{L^* + 1}(u_0)$, and by (\hyperlink{cr3}{CR.3}), after restricting $U$ to some sufficiently small neighborhood of $u_0$, both matrices have constant rank $p$. We will henceforth suppose $U$ has been replaced with any such restriction. As the rows of $\mathcal{M}_{L^*}$ are included in those of $\mathcal{M}_{L^*+1}$, we have for every $u \in U$ that $\textrm{ker}\,\mathcal{M}_{L^*}(u) = \textrm{ker}\, \mathcal{M}_{L^*+1}(u)$. Choose $p$ rows, $r^1(u)^\intercal, \ldots, r^p(u)^\intercal$ of $\mathcal{M}_{L^*}(u_0)$ that are linearly independent at $u_0$ (if $p = 0$ we instead choose the empty tuple). By possibly restricting $U$ further, we may suppose that $r^1(u)^\intercal, \ldots, r^p(u)^\intercal$ are linearly independent at every $u \in U$; we suppose this is true. Define:
    \[
        \hat{I}^a(X,u) = X \cdot r^a(u), \quad a = 1,\ldots, p.
    \]
    By \autoref{integrabilitycondlemma}, $\hat{I}^a \in \mathcal{I}^{L^*}$ for every $a = 1,\ldots, p$. Moreover, the Jacobian with respect to the $X$ variables of the vector valued function $(\hat{I}^1, \ldots, \hat{I}^p)$ is simply the matrix whose rows are $r^1(u)^\intercal, \ldots, r^p(u)^\intercal$, and hence has constant rank $p$.  Moreover, because $\mathcal{M}_{L^*}$ also has rank $p$ everywhere on $U$, these rows span the row-space of $\mathcal{M}_{L^*}(u)$ for all $u \in U$. Thus, again by \autoref{integrabilitycondlemma}, $\mathcal{Z}^{L^*} = \big\{(X,u) \in \mathbb{R}^d \times U : \mathcal{M}_{L^*}(u) X = 0\big\}$, and $\mathcal{Z}^{L^*} = \mathcal{Z}^{L^*+1}$. Thus by the theorem of \cite{veblen1926projective}, every initial condition $\tilde{X}_0 \in \textrm{ker} \mathcal{M}_{L^*} \cap \mathbb{R}^d_{++}$ determines a unique local solution about $u_0$. By \autoref{pdelemma}, this solution then corresponds to a local, common virtual numeraire for the subprofile $(\succsim^1,\ldots, \succsim^{d+1})$.

    $(\Longrightarrow):$ Suppose instead that some common, local virtual numeraire $\varphi$ exists.  Then by \autoref{pdelemma}, $X = \dot{\varphi}$ is a positive, local solution to \eqref{overdetpde2} on some neighborhood of $u_0$.

    For any $i \neq j$, we have $I^i_j\big(X(u), u\big) = \big(\mathcal{D}_j \psi^i_i\big)\big(X(u), u\big)$$=\frac{\partial}{\partial u^j} \bigg(\frac{\partial X^i}{\partial u^i} \bigg)$$=\frac{\partial}{\partial u^i} \bigg(\frac{\partial X^i}{\partial u^j} \bigg) = 0$. Hence every element of $\mathcal{I}^0$ vanishes along the graph of $X$. Now, suppose that for every $0 \le L \le \bar{L}$, every element of $\mathcal{I}^L$ vanishes on the graph of $X$. For any $I \in \mathcal{I}^L$ and $j = 1,\ldots, d$, $\big(\mathcal{D}_j I\big)\big(X(u), u\big) = \frac{\partial}{\partial u^j} I\big(X(u), u\big)$, as $\mathcal{D}_j$ corresponds to simply total differentiation of $I$ along the graph of $X$, and $I$ is constant on this graph. Because $\mathcal{I}^{\bar{L}+1}$ is the $U$-span of $\mathcal{I}^L$ and these derivatives, every element of $\mathcal{I}^{\bar{L}+1}$ must also vanish on the graph of $X$. By induction, we obtain that $\big(X(u), u\big) \in \mathcal{Z}^{L}$ for all $L \ge 0$. Now, by \autoref{integrabilitycondlemma}, for all $L \ge 0$ we also have $\mathcal{M}_L(u) X(u) = 0$. In particular, $X(u_0) \in \textrm{ker} \mathcal{M}_L(u_0) \cap \mathbb{R}^d_{++}$ as desired.
\end{proof}

\end{appendix}
\putbib
\end{bibunit}

\end{document}